\documentclass[letterpaper,onecolumn,11pt,accepted=2021-03-15]{quantumarticle}
\pdfoutput=1
\usepackage[utf8]{inputenc}
\usepackage[english]{babel}
\usepackage[T1]{fontenc}
\usepackage{amsmath}
\usepackage{hyperref}

\usepackage{tikz}
\usepackage{lipsum}

\usepackage{amssymb}
\usepackage{amsmath,mathtools}
\usepackage{amsfonts}
\usepackage{amsthm}
\usepackage{dsfont}
\usepackage{amssymb}
\usepackage{enumerate}

\usepackage{soul,xcolor}

\setstcolor{red}

\usepackage[numbers,sort&compress]{natbib}

\newcommand{\Tr}{\operatorname{Tr}}
\newcommand{\I}{\mathds{1}}
\renewcommand{\t}{{\scriptscriptstyle\mathsf{T}}}

\renewcommand{\dim}{\mathrm{dim}}

\newcommand\im{\operatorname{im}}

\newcommand{\CC}{\mathbb{C}}

\newcommand{\RR}{\mathbb{R}}

\newcommand{\calC}{\mathcal{C}}
\newcommand{\calV}{\mathcal{V}}
\newcommand{\calW}{\mathcal{W}}
\newcommand{\calX}{\mathcal{X}}
\newcommand{\calY}{\mathcal{Y}}
\newcommand{\calZ}{\mathcal{Z}}

\renewcommand\L{\mathrm{L}}
\newcommand\Pos{\mathrm{Pos}}
\newcommand\Pd{\mathrm{Pd}}
\newcommand\Herm{\mathrm{Herm}}
\newcommand\C{\mathrm{C}}
\newcommand\D{\mathrm{D}}

\newcommand\CP{\mathrm{CP}}

\newcommand{\F}{\operatorname{F}}

\newcommand\dom{\operatorname{dom}}

\newcommand\relint{\operatorname{relint}}

\newcommand{\tinyspace}{\mspace{1mu}}

\newcommand{\op}[1]{\operatorname{#1}}
\newcommand{\complex}{\mathbb{C}}

\newcommand{\norm}[1]{\lVert\tinyspace #1 \tinyspace\rVert}
\newcommand{\bignorm}[1]{\bigl\lVert\tinyspace #1 \tinyspace\bigr\rVert}
\newcommand{\Bignorm}[1]{\Bigl\lVert\tinyspace #1 \tinyspace\Bigr\rVert}
\newcommand{\biggnorm}[1]{\biggl\lVert\tinyspace #1 \tinyspace\biggr\rVert}

\theoremstyle{plain}
\newtheorem{theorem}{Theorem}
\newtheorem{lemma}[theorem]{Lemma}
\newtheorem{corollary}[theorem]{Corollary}
\newtheorem{proposition}[theorem]{Proposition}
\newtheorem{conjecture}[theorem]{Conjecture}

\theoremstyle{definition}

\newtheorem{remark}[theorem]{Remark}

\begin{document}

\title{%
  \makebox{Certifying optimality for convex quantum channel}
  \makebox{optimization problems}}

\author{Bryan Coutts}
\affiliation{Institute for Quantum Computing, University of Waterloo, Canada}
\affiliation{School of Computer Science, University of Waterloo, Canada}
\author{Mark Girard}
\orcid{0000-0003-4862-9492}
\affiliation{Institute for Quantum Computing, University of Waterloo, Canada}
\author{John Watrous}
\affiliation{Institute for Quantum Computing, University of Waterloo, Canada}
\affiliation{School of Computer Science, University of Waterloo, Canada}
\affiliation{Canadian Institute for Advanced Research, Toronto, Canada}
\orcid{0000-0002-4263-9393}
\maketitle

\begin{abstract}
  We identify necessary and sufficient conditions for a quantum channel to be
  optimal for any convex optimization problem in which the optimization is
  taken over the set of all quantum channels of a fixed size.
  Optimality conditions for convex optimization problems over the set of all
  quantum measurements of a given system having a fixed number of measurement
  outcomes are obtained as a special case.
  In the case of linear objective functions for measurement optimization
  problems, our conditions reduce to the well-known Holevo--Yuen--Kennedy--Lax
  measurement optimality conditions.
  We illustrate how our conditions can be applied to various state
  transformation problems having non-linear objective functions based on the
  fidelity, trace distance, and quantum relative entropy.
\end{abstract}

\section{Introduction}

Several problems and settings that arise in quantum information theory can
be expressed as \emph{optimization problems} in which a real-valued function,
defined for a class of quantum channels or measurements, is either minimized or
maximized.
The problem of \emph{minimum error quantum state discrimination}
\cite{BarnettC09}, in which a quantum state randomly selected from a known
ensemble of states is to be identified with the smallest possible probability
of error by means of a measurement, provides a well-known example.
This problem is naturally expressed as the optimization of a real-valued linear
function defined on the set of all measurements with a fixed number of
outcomes.
Other examples arise in the study of \emph{quantum cloning}
\cite{Bruss1997,ScaraniIGA05} and the closely related notion of
\emph{quantum money} \cite{AaronsonFGHKL12},
where one is generally interested in knowing how well an optimally selected
quantum channel can transform a single copy of a given state into multiple
copies of the same state, with respect to a number of different figures of
merit.
Another example can be found in quantum complexity theory, in which
\emph{two-message quantum interactive proof systems} \cite{JainUW09} are
naturally analyzed as optimization problems in which the objective function
describes the probability that a given verifier accepts, and where the
optimization is over all quantum channels of a fixed size, which describe the
possible actions of a prover.

Concerning the optimization of \emph{linear functions} defined on the set of
all measurements with a fixed number of outcomes, necessary and sufficient
conditions for optimality were identified by Holevo
\cite{Holevo73-Kyoto,Holevo73-statistical} and Yuen, Kennedy,
and Lax \cite{YuenKL70,YuenKL75}.
These conditions, which are described explicitly later in this paper, are
relatively easy to check;
the problem of actually finding or approximating an optimal measurement, while
efficiently solvable through the use of semidefinite programming
\cite{JezekRF02,Ip03,EldarMV03}, is in general a more computationally involved
task.
These optimality conditions can be easily extended to obtain optimality
conditions for real-valued linear functions defined on the set of all quantum
channels transforming one quantum system to another.

We prove a generalization of these results to convex optimization problems
whose objective functions are not necessarily linear.
To be more precise, we consider optimization problems of the form
\begin{equation}
  \label{eq:cvxchannels}
  \begin{aligned}
    \text{minimize}\;\; & f(\Phi)\\
    \text{subject to}\;\; & \Phi\in\C(\calX,\calY),
  \end{aligned}
\end{equation}
where
\begin{equation}
  f:\C(\calX,\calY)\rightarrow\RR\cup\{\infty\}
\end{equation}
is a convex function.
Here, $\C(\calX,\calY)$ denotes the set of all channels (i.e., completely
positive and trace-preserving linear maps) from an input system to an output
system having associated complex Euclidean spaces $\calX$ and $\calY$,
respectively.
A channel $\Phi\in\C(\calX,\calY)$ is said to be \emph{optimal} for the
problem \eqref{eq:cvxchannels} if it is the case that $f(\Phi) \leq f(\Psi)$
for all $\Psi\in\C(\calX,\calY)$.
In this paper we do not consider the difficulty of finding or approximating
an optimal channel $\Phi$ for a given function $f$, but instead we focus only
on the task of \emph{verifying} that a given channel $\Phi$ is indeed optimal.
The optimality conditions we obtain can be easily checked for differentiable
functions $f$, and can also be used to verify optimality of channels for
some non-differentiable functions.
We stress that our optimality conditions are not generic optimality conditions
that hold for all convex optimization problems, but rather rely on a specific
structure that arises when the optimization is over all quantum channels of a
fixed size.

There are, of course, situations in which one would prefer a method to find an
optimal channel $\Phi$ for a chosen function $f$, as opposed to simply verifying
the optimality of a given $\Phi$, but the task of verifying optimality
nevertheless has value for multiple reasons.
For instance, one might hypothesize that a particular channel $\Phi$ is optimal
based upon an intuition concerning the function $f$, or upon a heuristic
method, making the task of verifying optimality essentially important.
The computational task of finding an optimal solution for a chosen function
$f$ might also be expensive, time-consuming, or delegated to an untrusted
computer, but once this task has been performed the optimality of the solution
can be verified, allowing anyone who performs the verification to trust in
the optimality of the solution.
Finally, there are situations in which the function $f$ could be indeterminate
in some respect, eliminating the possibility that a numerical computation
could reveal an optimal solution, but potentially allowing for an optimal
solution to be expressed and checked analytically.
(The results of Bacon, Childs, and van Dam \cite{BaconCvD05} on hidden subgroup
algorithms for certain groups provide a striking example of this potential.)

Our optimality conditions are applicable to convex optimization problems in
which the optimization is over all measurements having a fixed number of
outcomes, as opposed to being over all channels of a fixed size.
This is done through a standard correspondence between measurements and
\emph{quantum-to-classical channels}, described in the section following this
one.
We observe that Holevo \cite{Holevo73-Kyoto,Holevo73-statistical} also derived
optimality conditions for optimizations over measurements having differentiable
(but not necessarily convex) objective functions.
Holevo proved that these conditions are necessary for optimality, but did not
prove they are sufficient (as they are not sufficient in general).
When one restricts their attention to convex objective functions, our
optimality conditions are equivalent to a set of intermediate conditions
identified by Holevo, but not to the final set of conditions he identified.

We provide a few examples of how our optimality conditions can be applied to
interesting categories of optimization problems.
As a simple warm-up, we first explain how our conditions imply the
Holevo--Yuen--Kennedy--Lax conditions for the optimality of measurements for
minimum error state discrimination, which are easily extended to channel
optimization problems having linear objective functions.
We then discuss optimization problems relating to state transformations having
objective functions based on fidelity, trace distance, and quantum relative
entropy.
We note that optimization problems of the form we consider have recently been
studied in the context of \emph{recovery measures} (i.e., the
\emph{fidelity of recovery} and generalizations of this measure)
\cite{FawziR15,SeshadreesanW15,BrandaoHOS15,CooneyHMOSWW16,BertaT16,BertaFT17}.
Expositions of this topic can be found in \cite{Sutter18} and Chapter 12 of
\cite{Wilde17}.

\section{Background and notation}
\label{sec:cvx}\label{sec:preliminaries}

This section summarizes some concepts from convex analysis and optimization
theory, narrowly focused on their applications to this paper.
Further information on these topics can be found in \cite{Rockafellar1970},
\cite{Borwein2006}, \cite{Boyd2004}, and \cite{Mordukhovich2013}, for
instance.
We assume the reader is familiar with quantum information theory, which is
covered in the books \cite{NielsenC00}, \cite{Wilde17}, and \cite{watrous2018},
among others.
We will, however, summarize the notion of the Choi operator of a channel,
clarify its basic connection to the sorts of optimization problems we consider,
and discuss the correspondence between measurements and quantum-to-classical
channels, as these topics are essential to an explanation of our results.
It will also be helpful to begin the section by establishing some basic
notation and terminology concerning linear algebra.

\subsubsection*{Linear algebra notation and terminology}

When we refer to a \emph{complex Euclidean space}, we mean $\CC^n$ for some
positive integer $n$, or more generally the complex vector space consisting of
vectors indexed by an arbitrary finite set in place of the index set
$\{1,\ldots,n\}$.
The elementary unit vectors of the space $\CC^n$ are denoted
$e_1,\ldots,e_n$.
Complex Euclidean spaces will be denoted by capital calligraphic letters such
as $\calX$, $\calY$, and $\calZ$.

For a complex Euclidean space $\calX$, the space of linear operators on $\calX$
is denoted $\L(\calX)$, and the identity operator on $\calX$ is denoted
$\I_{\calX}$.
The adjoint (or conjugate-transpose) of an operator $X$ is denoted $X^{\ast}$.
For indices $j,k\in\{1,\ldots,n\}$, the operator $E_{j,k}\in\L(\CC^n)$ is
defined as $E_{j,k} = e_j e_k^{\ast}$.
Equivalently, with respect to the basis $\{e_1,\ldots,e_n\}$, the matrix
representation of $E_{j,k}$ has a 1 in the $(j,k)$ entry, with all other
entries~0.

The real vector space of Hermitian operators acting on a complex Euclidean
space $\calX$ is denoted $\Herm(\calX)$; the cone of positive semidefinite
operators acting on $\calX$ is denoted $\Pos(\calX)$; the set of positive
definite operators acting on $\calX$ is denoted $\Pd(\calX)$; and the set of
density operators (i.e., positive semidefinite operators having unit trace) is
denoted $\D(\calX)$.
The Hilbert--Schmidt inner product of two Hermitian operators
$X,Y\in \Herm(\calX)$ is given by
\begin{equation}
  \langle X,Y\rangle = \Tr(X Y).
\end{equation}
For a subspace $\calV\subseteq\calX$ of a complex Euclidean space $\calX$, we
write $\Pi_{\calV}\in\Pos(\calX)$ to denote the (orthogonal) projection
operator that projects onto the subspace $\calV$.
Finally, whenever we refer to the \emph{inverse} of a positive semidefinite
operator $X\in\Pos(\calX)$, it should be understood that we are referring to
the \emph{Moore--Penrose pseudo-inverse} of $X$ (i.e., the operator that acts as
the inverse of $X$ on the image of $X$ and zero on the kernel of $X$).

\subsubsection*{Convex functions taking real or infinite values}

Let $\calX$ be a complex Euclidean space and let
\begin{equation}
  \label{eq:possibly-infinite-function}
  f:\Herm(\calX)\rightarrow\RR\cup\{\infty\}
\end{equation}
be a function mapping each Hermitian operator to either a real number or to
positive infinity.
The \emph{domain} of $f$ is defined as
\begin{equation}
  \dom(f)=\{X\in\Herm(\calX) : f(X)\in\RR\}.
\end{equation}
For any function of the form $f:\calC\rightarrow\RR$ defined only on a
subset $\calC\subseteq\Herm(\calX)$, one may naturally extend $f$ to a
function of the form \eqref{eq:possibly-infinite-function} by defining
$f(X)=\infty$ whenever $X\not\in\calC$.

A function $f:\Herm(\calX)\rightarrow\RR\cup\{\infty\}$ is \emph{convex}
if $\dom(f)$ is a convex set and $f$ is convex on $\dom(f)$.
More explicitly, $f$ is convex if for all $X,Y\in\dom(f)$ and
$\lambda \in [0,1]$, it is the case that
$\lambda X + (1 - \lambda) Y\in \dom(f)$ and
\begin{equation}
  f(\lambda X + (1 - \lambda)Y) \leq
  \lambda f(X) + (1-\lambda) f(Y).
\end{equation}
A function $f$ of the form \eqref{eq:possibly-infinite-function} is
\emph{proper} if $\dom(f)\neq\varnothing$.

The \emph{indicator function} of a set $\calC\subseteq\Herm(\calX)$ is the
function
\begin{equation}
  I_\calC:\Herm(\calX)\rightarrow\RR\cup\{\infty\}
\end{equation}
defined as
\begin{equation}
  I_{\calC}(X) =
  \begin{cases}
    0       & X\in\calC\\
    \infty & X\notin\calC
  \end{cases}
\end{equation}
for all $X\in\Herm(\calX)$.
It is evident that $\dom(I_\calC)=\calC$ for every set
$\calC\subseteq\Herm(\calX)$, and if $\calC$ is a convex set then $I_{\calC}$
is a convex function.

\subsubsection*{Subdifferentials}

Let $f$ be a proper function of the form \eqref{eq:possibly-infinite-function}
and let $X\in\dom(f)$.
The \emph{subdifferential} of $f$ at $X$ is the set defined as
\begin{equation}
  \partial f(X) = \{Z\in\Herm(\calX) \,:\,
  f(Y)-f(X) \geq \langle Z,Y-X\rangle\;\text{for all}\;
  Y\in\dom(f)\}.
\end{equation}

A key property of the subdifferential of a proper function, which follows
trivially from the definition of the subdifferential, is its relation to global
minima:
$X\in\dom(f)$ is a global minimizer of $f$ if and only if $0\in\partial f(X)$.
Two additional properties of subdifferentials that are relevant to this paper
are the following:
\begin{enumerate}
  \item
    If $f$ is convex and differentiable at $X\in\dom(f)$ then
    $\partial f(X)=\{\nabla f(X)\}$. 
  \item
    If $\norm{\cdot}$ is any norm on $\Herm(\calX)$ and $f(X) = \norm{X}$
    for all $X\in\Herm(\calX)$, then 
    \begin{equation}
      \partial f(X) =
      \{Y :  \langle Y,X\rangle = \lVert X\rVert,\, \lVert Y\rVert_*\leq 1\},
    \end{equation}
    where $\norm{Y}_{\ast} = \sup\{\langle Y,X\rangle : \norm{X}\leq 1\}$ is
    the dual norm to $\norm{\cdot}$. (See Example 14.4.2 in \cite{Lange2013}.)
\end{enumerate}
We will make use of the following proposition, which presents a variant of
the \emph{chain rule for subdifferentials}.
In the statement of this proposition, $\relint$ denotes the
\emph{relative interior} of a set, and for a linear map
$\Lambda:\Herm(\calY)\rightarrow\Herm(\calX)$, the map
\begin{equation}
  \Lambda^{\ast}:\Herm(\calX)\rightarrow\Herm(\calY)
\end{equation}
denotes the adjoint map to $\Lambda$, which is the uniquely determined linear
map that satisfies
\begin{equation}
  \langle X, \Lambda(Y) \rangle = \langle \Lambda^{\ast}(X), Y \rangle
\end{equation}
for all $X\in\Herm(\calX)$ and $Y\in\Herm(\calY)$.
While the following result can be easily derived from basic rules of convex
analysis, we note that this particular result is stated and proved as
Theorem~7.2 in \cite{Mordukhovich2015}.

\begin{proposition}
  \label{prop:subdifferential-chain-rule}
  Let $\calX$ and $\calY$ be complex Euclidean spaces, let
  $f:\Herm(\calX)\rightarrow\RR\cup\{\infty\}$ be a convex function,
  and let $g:\Herm(\calY)\rightarrow\Herm(\calX)$ be an affine linear map,
  meaning that
  \begin{equation}
    g(Y) = \Lambda(Y) + A
  \end{equation}
  for all $Y\in\Herm(\calY)$, for some choice of a linear map
  $\Lambda:\Herm(\calY)\rightarrow\Herm(\calX)$ and an operator
  $A\in\Herm(\calX)$.
  \begin{enumerate}
  \item
    For every $Y\in\Herm(\calY)$, it is the case that
    \begin{equation}
      \partial (f\circ g)(Y) \supseteq \Lambda^{\ast}(\partial f
      (g(Y))).
    \end{equation}
  \item
    If $\im(\Lambda) \cap \relint(\dom(f)) \not= \varnothing$, then
    \begin{equation}
      \partial (f\circ g)(Y) = \Lambda^{\ast}(\partial f (g(Y)))
    \end{equation}
    for every $Y\in\Herm(\calY)$.
  \end{enumerate}
\end{proposition}

Lastly, we present the subdifferentials of some indicator functions.
The subdifferential of the indicator function of a set
$\mathcal{C}\subseteq\Herm(\calX)$ at an operator $X\in\Herm(\calX)$ may be
expressed as
\begin{equation}\label{eq:normalcone}
  \partial I_{\mathcal{C}}(X) = \begin{cases}
                                 \{Z\in\Herm(\calX) \,:\,
  \langle Z,X\rangle \geq \langle Z,Y\rangle\;\text{for all}\;
  Y\in\mathcal{C}\} & \text{if }X\in\mathcal{C}\\
  \varnothing & \text{otherwise.}
                                \end{cases}
\end{equation}
The set described in \eqref{eq:normalcone} is sometimes referred to as the
\emph{normal cone} of $\mathcal{C}$ at $X$.
(For further discussion of normal cones and subdifferentials of indicator
functions, see Section 2.2 in \cite{Mordukhovich2013} and Chapter 23 in
\cite{Rockafellar1970}.)
For a positive semidefinite operator $X\in\Pos(\calX)$ it is straightforward to
verify that
\begin{equation}\label{eq:subdiff_indicator_pos}
  \partial I_{\Pos(\calX)}(X)
  = \{Z\,:\, -Z\in\Pos(\calX),\, \langle Z,X\rangle=0\}
  = \{Z\,:\, -Z\in\Pos(\calX),\, ZX=0\}.
\end{equation}
Let $\Lambda:\Herm(\calX)\to\Herm(\calY)$ be a linear map, let
$A\in\Herm(\calY)$ be an operator, and define the set
\begin{equation}
 \mathcal{K}=\{X\in\Herm(\calX)\,:\, \Lambda(X)=A\}.
\end{equation}
(This is the pre-image of $A$ under $\Lambda$.)
For an operator $X\in\mathcal{K}$, it holds that  
\begin{equation}\label{eq:subdiff_indicator_preimage}
  \partial I_{\mathcal{K}}(X)= \{Z\in\Herm(\calX)\, :\,
  \langle Z,Y\rangle = 0 \;\text{for all}\;Y\in\ker(\Lambda)\} = \im(\Lambda^*).
\end{equation}
(See for example Proposition 2.12 in \cite{Mordukhovich2013}.)

\subsubsection*{Convex optimization problems}

Convex optimization problems in quantum information theory often have the
following general form, for some choice of a convex function
$f:\Herm(\calX)\rightarrow\RR\cup\{\infty\}$ and a convex set
$\calC\subseteq\Herm(\calX)$:
\begin{equation}
  \label{eq:cvxgeneral}
  \begin{aligned}
    &\text{minimize}   && f(X)\\
    &\text{subject to} && X\in\calC.
  \end{aligned}
\end{equation}
An operator $X\in\calC\cap\dom(f)$ is said to be \emph{optimal} for the
optimization problem \eqref{eq:cvxgeneral} if $f(X)\leq f(Y)$ for all
$Y\in\calC$.
As we will see below, convenient conditions for optimality can be given if it
is the case that
\begin{equation}
  \label{eq:relintdomfrelintC}
  \relint\bigl(\dom(f)\bigr) \cap \relint (\calC) \neq \varnothing.
\end{equation}

A (constrained) convex optimization problem of the form \eqref{eq:cvxgeneral}
can be considered as an unconstrained convex optimization problem by
minimizing $f(X) + I_{\calC}(X)$ over all Hermitian operators
$X\in\Herm(\calX)$.
The domain of the function $f + I_{\calC}$ is given by
\begin{equation}
  \dom(f+I_\calC) = \dom(f)\cap\calC.
\end{equation}
As was mentioned previously, an operator $X\in\dom(f)\cap\calC$ is a global
minimizer of $f(X) + I_{\calC}(X)$ if and only if
$0\in\partial(f+I_{\calC})(X)$.
If the condition in \eqref{eq:relintdomfrelintC} holds and both $f$ and $\calC$
are convex, a subdifferential sum-rule (see, e.g.\ Theorem 23.8 of
\cite{Rockafellar1970}) implies that
\begin{equation}
  \partial(f+I_{\calC})(X) = \partial f(X) + \partial I_{\calC}(X)
\end{equation}
for all $X\in\dom(f)\cap\calC$.
For this reason, a characterization of the subdifferential set
$\partial I_{\calC}(X)$ for a convex set $\calC$ can be useful in
identifying necessary and sufficient conditions for optimality.

\subsubsection*{Optimizing over Choi operators of channels}

For complex Euclidean spaces $\calX$ and $\calY$, the set of completely
positive, trace-preserving linear maps (i.e., channels) from $\L(\calX)$ to
$\L(\calY)$ is denoted $\C(\calX,\calY)$.
Given a linear map $\Phi:\L(\calX)\rightarrow\L(\calY)$, its Choi
representation $J(\Phi)\in\L(\calY\otimes\calX)$ is defined as
\begin{equation}
 J(\Phi) = \sum_{j,k=1}^{n} \Phi(E_{j,k})\otimes E_{j,k},
\end{equation}
assuming $\calX = \CC^n$.
(An analogous definition is used for index sets other than
$\{1,\ldots,n\}$.)
Through this representation, the set of channels is isomorphic to the set 
\begin{equation}
  \label{eq:Choichannels}
  J(\C(\calX,\calY)) = \{X\in\Pos(\calY\otimes\calX)\,:\,
  \Tr_{\calY}(X)=\I_{\calX}\}\subseteq\Herm(\calY\otimes\calX).
\end{equation}
This set is convex, and it is helpful to observe that its relative interior is
\begin{equation}
  \relint\bigl(J(\C(\calX,\calY))\bigr) =
  \{X\in\Pd(\calY\otimes\calX)\,:\,\Tr_{\calY}(X)=\I_{\calX}\}.
\end{equation}
The action of a map $\Phi$ can be recovered from its Choi representation
by the relation
\begin{equation}
  \Phi(X)=\Tr_{\calX}((\I_{\calY}\otimes X^\t)J(\Phi))
\end{equation}
for all operators $X\in\L(\calX)$, where $X^\t$ denotes the transpose of~$X$.

An optimization problem of the form
\begin{equation}
  \label{eq:cvxchannels-g}
  \begin{aligned}
    \text{minimize}\;\; & g(\Phi)\\
    \text{subject to}\;\; & \Phi\in\C(\calX,\calY),
  \end{aligned}
\end{equation}
where
\begin{equation}
  g:\C(\calX,\calY)\rightarrow\RR\cup\{\infty\}
\end{equation}
is a given function, can equivalently be expressed as
\begin{equation}
  \label{eq:cvxchannels-Choi}
  \begin{aligned}
    \text{minimize}\;\; & f(X)\\
    \text{subject to}\;\; & X\in J(\C(\calX,\calY)),
  \end{aligned}
\end{equation}
where
\begin{equation}
  f:J(\C(\calX,\calY))\rightarrow\RR\cup\{\infty\}
\end{equation}
is defined as $f(J(\Phi)) = g(\Phi)$ for all $\Phi\in\C(\calX,\calY)$.
Although the formulations \eqref{eq:cvxchannels-g} and
\eqref{eq:cvxchannels-Choi} are equivalent, it will be convenient for us to
focus primarily on the formulation \eqref{eq:cvxchannels-Choi}.
The results we obtain can, however, easily be adapted to the formulation
\eqref{eq:cvxchannels-g}.

\subsubsection*{Subdifferentials of the indicator function of the set of
  Choi operators of channels}

The following proposition, which follows from a straightforward application of
the rules of subdifferentials, provides a useful characterization of the
subdifferential of the indicator function of $J(\C(\calX,\calY))$ at every
point $X\in J(\C(\calX,\calY))$.
While we suspect the result is not new, we are not aware of a reference for it,
and have included its proof for completeness.

\begin{proposition}
  \label{prop:Choi-subdifferential}
  For complex Euclidean spaces $\calX$ and $\calY$, and for $X\in
  J(\C(\calX,\calY))$, it holds that
  \begin{equation}
    \label{eq:subdiffJCXY}
    \partial I_{J(\C(\calX,\calY))}(X) =
    \{-Y - \I_{\calY}\otimes Z \,:\,
    Z\in\Herm(\calX),\;
    Y\in\Pos(\calY\otimes\calX),\;
    YX = 0 \}.
  \end{equation}
\end{proposition}

\begin{proof}
  Define
  $\mathcal{K}=\{K\in\Herm(\calY\otimes\calX)\,:\,\Tr_{\calY}(K)=\I_{\calX}\}$.
  It is the case that
  \begin{equation}
    J(\C(\calX,\calY)) = \Pos(\calY\otimes\calX) \cap \mathcal{K},
  \end{equation}
  and therefore
  \begin{equation}
    I_{J(\C(\calX,\calY))}
    = I_{\Pos(\calY\otimes\calX)} + I_{\mathcal{K}}.
  \end{equation}
  As $\Pos(\calY\otimes\calX)$ and $\mathcal{K}$ are both convex and
  \begin{equation}
    \relint(\Pos(\calY\otimes\calX)) \cap \relint(\mathcal{K})
    = \{P\in\Pd(\calY\otimes\calX)\,:\,\Tr_{\calY}(P) = \I_{\calX}\}
    \not= \varnothing,
  \end{equation}
  one has that
  \begin{equation}
    \partial I_{J(\C(\calX,\calY))}(X)
    = \partial I_{\Pos(\calY\otimes\calX)}(X) + \partial
    I_{\mathcal{K}}(X).
  \end{equation}
  Finally, from \eqref{eq:subdiff_indicator_pos} it holds that
  \begin{equation}
    \partial I_{\Pos(\calY\otimes\calX)}(X)
    = \{-Y\,:\,Y\in\Pos(\calY\otimes\calX),\; XY = 0\}
  \end{equation}
  and from \eqref{eq:subdiff_indicator_preimage} it holds that
  \begin{equation}
    \partial I_{\mathcal{K}}(X) = \{
    \I_{\calY}\otimes Z\,:\,Z\in\Herm(\calX)\}.
  \end{equation}
  This implies the proposition.  
\end{proof}

\subsubsection*{The Lagrange dual problem for channel optimization}

Consider a channel optimization problem expressed in the following form
that is equivalent to \eqref{eq:cvxchannels-Choi}:
\begin{equation}
  \label{eq:primal-problem}
  \begin{aligned}
    \text{minimize}\;\; & f(X)\\
    \text{subject to}\;\;
    & \Tr_{\calY}(X) = \I_{\calX},\\
    & X\in\Pos(\calY\otimes\calX).
  \end{aligned}
\end{equation}
One may then formulate the associated Lagrange dual problem:
\begin{equation}
  \label{eq:Lagrange-dual-problem}
  \begin{aligned}
    \text{maximize}\;\; & g(Y,Z)\\
    \text{subject to}\;\;
    & Y\in\Pos(\calY\otimes\calX),\\
    & Z \in \Herm(\calX),
  \end{aligned}
\end{equation}
where
\begin{equation}
  g(Y,Z) = \Tr(Z) + \inf_{X\in\Herm(\calY\otimes\calX)}
  \bigl( f(X) - \bigl\langle X,Y + \I_{\calY}\otimes Z\bigr\rangle \bigr)
\end{equation}
for all $Y\in\Pos(\calY\otimes\calX)$ and $Z\in\Herm(\calX)$.

The optimal value of the Lagrange dual problem \eqref{eq:Lagrange-dual-problem}
is a lower-bound for the optimal value of the original problem
\eqref{eq:primal-problem} (even if the function $f$ is not convex), which is a
property known as weak duality.
\emph{Slater's theorem} implies that if $f$ is convex, and there exists a
positive definite operator $X\in\Pd(\calY\otimes\calX)$ such that
$\Tr_{\calY}(X)=\I_{\calX}$ and $X\in\relint(\dom(f))$, then the problems
\eqref{eq:primal-problem} and \eqref{eq:Lagrange-dual-problem} have the same
optimal value---a property known as \emph{strong duality}.
(A proof of Slater's theorem can be found in Section 5.3.2 of \cite{Boyd2004}.)

\subsubsection*{Quantum measurements as channels}

Optimizations over the set of all measurements having a fixed number of
outcomes can be expressed as optimizations over channels, as we now explain.
Consider first a measurement on a complex Euclidean space $\calX$ that has $m$
possible measurement outcomes, and is described by measurement operators
$\{P_1,\ldots,P_m\}\subseteq\Pos(\calX)$.
Letting $\calY = \CC^m$, this measurement can be represented by the channel
$\Phi\in\C(\calX,\calY)$ defined as
\begin{equation}
  \label{eq:quantum-to-classical-channel}
  \Phi(X) = \sum_{k=1}^m \langle P_k , X \rangle \, E_{k,k}
\end{equation}
for all $X\in\L(\calX)$.
Any channel expressible in this form is called a
\emph{quantum-to-classical channel}.
An equivalent condition to a channel taking the form
\eqref{eq:quantum-to-classical-channel} is that its Choi operator takes the
form 
\begin{equation}
  \label{eq:quantum-to-classical-Choi}
  J(\Phi) = \sum_{k = 1}^m E_{k,k} \otimes P_k^{\t}.
\end{equation}

Next, for each $k\in\{1,\ldots,m\}$, define a linear map
$\Xi_k : \Herm(\calY\otimes\calX) \rightarrow \Herm(\calX)$ as
\begin{equation}
  \Xi_k(X) = \bigl( (e_k^{\ast} \otimes\I_{\calX}) X
  (e_k \otimes\I_{\calX}) \bigr)^{\t},
\end{equation}
and observe that one has
\begin{equation}
  \label{eq:measurement-from-Choi}
  P_k = \Xi_k(J(\Phi))
\end{equation}
for the quantum-to-classical channel $\Phi$ given by
\eqref{eq:quantum-to-classical-channel}.
In words, for a measurement represented by operators $\{P_1,\ldots,P_m\}$, the
linear map $\Xi_k$ allows for the recovery of the operator $P_k$ from the Choi
operator $J(\Phi)$ of the quantum-to-classical channel associated with that
measurement.
A function $g(P_1,\ldots,P_m)$ of these measurement operators can therefore
be expressed as a function
\begin{equation}
  \label{eq:g-from-f}
  f(J(\Phi)) = g\bigl(\Xi_1(J(\Phi)),\ldots,\Xi_m(J(\Phi))\bigr)
\end{equation}
of the Choi operator $J(\Phi)$.
Note that if $\Phi\in\C(\calX,\calY)$ is an arbitrary (i.e., not necessarily
quantum-to-classical) channel, then the operators $P_1,\ldots,P_m$ defined by
\eqref{eq:measurement-from-Choi} will still necessarily satisfy $P_1,\ldots,P_m
\in \Pos(\calX)$ and $P_1+\cdots+P_m = \I_{\calX}$, and therefore represent a
valid measurement.
(In general, the same measurement is given by a continuum of channels $\Phi$,
including the quantum-to-classical channel described earlier.)
An optimization problem of the form
\begin{equation}
  \label{eq:optimize-over-measurements}
  \begin{aligned}
    & \text{minimize}   && g(P_1,\ldots,P_m)\\[1mm]
    & \text{subject to} && P_1,\ldots,P_m \in \Pos(\calX)\\
    & && P_1+\cdots+P_m = \I_{\calX}
  \end{aligned}
\end{equation}
can therefore be expressed equivalently as follows:
\begin{equation}
  \begin{aligned}
    & \text{minimize}   && f(X)\\[1mm]
    & \text{subject to} && X\in J(\C(\calX,\calY))
  \end{aligned}
\end{equation}
for $f$ defined from $g$ as in \eqref{eq:g-from-f}.

\section{Optimality conditions for convex channel optimization}

In this section, we present our main general result regarding optimality
conditions for convex optimization problems over quantum channels.
As suggested in the previous section, it is convenient to associate quantum
channels with their Choi representations, and to consider optimization problems
of the form
\begin{equation}\label{eq:cvxchannelschoi}
  \begin{aligned}
    & \text{minimize}   && f(X)\\
    & \text{subject to} && X\in J(\C(\calX,\calY))
  \end{aligned}
\end{equation}
for various choices of convex functions
\begin{equation}
  f:\Herm(\calY\otimes\calX)\rightarrow\RR\cup\{\infty\}.
\end{equation}
Optimality conditions for such problems can be translated to optimality
conditions for problems of the form \eqref{eq:cvxchannels}, as will be
illustrated in the section following this one.

\begin{theorem}\label{thm:cvxchannelschoigeneral}
  Let $f:\Herm(\calY\otimes\calX)\rightarrow\RR\cup\{\infty\}$ be a convex
  function such that
  \begin{equation}
    \relint\bigl(\dom(f)\bigr)\cap
    \relint\bigl(J\bigl(\C(\calX,\calY)\bigr)\bigr) \not= \varnothing,
  \end{equation}
  and let $X\in J\bigl(\C(\calX,\calY)\bigr)$ be the Choi representation of a
  channel such that $X\in\dom(f)$.
  The following statements are equivalent:
  \begin{enumerate}
  \item
    The operator $X$ is optimal for the optimization problem
    \eqref{eq:cvxchannelschoi}.
  \item
    There exists an operator $H\in\partial f(X)$ satisfying
    \begin{equation}\label{eq:TrHXHIHX}
      \Tr_{\calY}(HX)\in\Herm(\calX)
      \quad \text{and}\quad
      H\geq \I_{\calY}\otimes\Tr_{\calY}(HX).
    \end{equation}
  \end{enumerate}
\end{theorem}
  
\begin{proof}
  An operator $X$ is optimal for the problem \eqref{eq:cvxchannelschoi} if and
  only if
  \begin{equation}
    0\in\partial(f+I_{J(\C(\calX,\calY))})(X)
    = \partial f(X) + \partial I_{J(\C(\calX,\calY))}(X).
  \end{equation}
  By the characterization of the subdifferential of the indicator function
  $I_{J(\C(\calX,\calY))}$ given by
  Proposition~\ref{prop:Choi-subdifferential}, it follows that $X$ is optimal
  for the optimization problem \eqref{eq:cvxchannelschoi} if and only if there
  exist operators $Y\in\Pos(\calY\otimes\calX)$, $Z\in\Herm(\calX)$, and
  $H\in\partial f(X)$ satisfying
  \begin{equation}
    \label{eq:YX=0HYIX0}
    YX = 0
    \quad\text{and}\quad
    H-Y-\I_{\calY}\otimes Z =0.
  \end{equation}

  Assume first that statement 2 holds, and choose
  $Y\in\Pos(\calY\otimes\calX)$ and $Z\in\Herm(\calX)$ as
  \begin{equation}
    Y = H-\I_{\calY}\otimes Z
    \quad\text{and}\quad
    Z = \Tr_{\calY}(HX).
  \end{equation}
  As $X$ and $Y$ are both positive operators and 
  \begin{equation}
    \langle Y,X\rangle
    = \langle H - \I_{\calY}\otimes \Tr_{\calY}(HX), X\rangle
    = \langle H,X\rangle -\langle \Tr_{\calY}(HX),\I_{\calX}\rangle
    = 0,
  \end{equation}
  it follows that $YX=0$.
  The conditions in \eqref{eq:YX=0HYIX0} are therefore satisfied, implying that
  statement~1 holds.

  Now assume that statement~1 holds, so that there exist operators
  $Y\in\Pos(\calY\otimes\calX)$, $Z\in\Herm(\calX)$, and $H\in\partial f(X)$
  such that $YX = 0$ and $H-Y-\I_{\calY}\otimes Z =0$.
  It follows that $HX = (\I_{\calY}\otimes Z)X$ and therefore
  \begin{equation}
    \Tr_{\calY}(HX) = \Tr_{\calY}\bigl((\I_{\calY}\otimes Z)X\bigr)
    = Z\Tr_{\calY}(X) = Z,
  \end{equation}
  which is a Hermitian operator.
  Moreover, one has $H-\I_{\calY}\otimes\Tr_{\calY}(HX) = Y$, which is positive
  semidefinite by assumption.
  The conditions in \eqref{eq:TrHXHIHX} are therefore satisfied, which implies
  that statement~2 holds, completing the proof.
\end{proof}

To make use of Theorem \ref{thm:cvxchannelschoigeneral}, one requires that the
relative interior of the domain of the objective function contains at least one
point in the relative interior of the set of channels.
This requirement is precisely Slater's condition for this optimization problem,
which guarantees strong duality with the corresponding dual problem
\eqref{eq:Lagrange-dual-problem}.
This regularity condition is automatically satisfied if $f$ is continuous at at
least one point in the set of Choi representations of channels.
In particular, if $f$ is differentiable at $X$ then one may take the operator
$H$ in Theorem \ref{thm:cvxchannelschoigeneral} to be $H=\nabla f(X)$.

\begin{corollary}\label{cor:cvxchannels}
  Let $f:\Herm(\calY\otimes\calX)\rightarrow\RR\cup\{\infty\}$ be a convex
  function, let $X\in J(\C(\calX,\calY))$ be the Choi representation of a
  channel, and assume that $f$ is differentiable at $X$.
  The operator $X$ is an optimal solution to the optimization problem
  \eqref{eq:cvxchannelschoi} if and only if
  \begin{equation}
    \Tr_{\calY}(\nabla f(X)X)\in\Herm(\calX) \quad\text{and}\quad
    \nabla f(X)\geq \I_{\calY}\otimes\Tr_{\calY}(\nabla f(X)X).
  \end{equation}
\end{corollary}

\begin{proof}
  As $f$ is differentiable at $X$, one has that
  $\partial f(X)=\{\nabla f(X)\}$.
  Furthermore, $f$ must be finite in some neighborhood around $X$ and thus
  $\relint(\dom(f))\cap\relint(J(\C(\calX,\calY)))\neq\varnothing$ must hold.
  The result now follows from Theorem \ref{thm:cvxchannelschoigeneral}. 
\end{proof}

We remark that the optimality conditions represented by this corollary
appear to be a special feature of problems involving an optimization over
channels.
In essence, for differentiable convex quantum channel optimization problems,
every optimal primal solution~$X$ determines an optimal dual solution
consisting of operators
$Z=\Tr_{\calY}(\nabla f(X) X)$ and
$Y=\nabla f(X) - \I_{\calY}\otimes\Tr_{\calY}(\nabla f(X) X)$. 

It is natural to ask if there is an approximate version of the implication that
statement~2 implies statement 1 in Theorem~\ref{thm:cvxchannelschoigeneral}.
That is, if the requirements \eqref{eq:TrHXHIHX} hold approximately for some
$H\in\partial f(X)$, then is $X$ necessarily close to being optimal?
The following theorem demonstrates that this is indeed the case, which allows
for bounds to be placed on the optimal value of such optimization problems in
the case the conditions in~\eqref{eq:TrHXHIHX} cannot be verified exactly.

\begin{theorem}
\label{thm:approxoptimal}
  Let $f:\Herm(\calY\otimes\calX)\rightarrow\RR\cup\{\infty\}$ be a function,
  let $\Phi\in \C(\calX,\calY)$ be a channel such that $J(\Phi)\in\dom(f)$,
  and let $H \in \partial f(X)$.
  It is the case that
  \begin{equation}\label{eq:epsbound}
    f(J(\Phi)) \leq \inf_{\Psi\in\C(\calX,\calY)} f(J(\Psi))
    + \varepsilon\,\dim(\calX),
  \end{equation}
  where
  \begin{equation}\label{eq:eps}
    \varepsilon = \inf_{P\in\Pos(\calY\otimes\calX)}
    \bignorm{H - \I_{\calY}\otimes \Tr_{\calY}(H J(\Phi)) - P}_{\infty}.
  \end{equation}
\end{theorem}

\begin{proof}
  Because the spectral norm is a continuous function, a compactness argument
  implies that there must exist a positive semidefinite operator
  $P\in\Pos(\calY\otimes\calX)$ for which
  \begin{equation}
    \varepsilon = \bignorm{H-\I_{\calY}\otimes\Tr_{\calY}(H J(\Phi)) -
      P}_{\infty}.
  \end{equation}
  We will let such a $P$ be fixed for the remainder of the proof.
  
  Define a Hermitian operator
  \begin{equation}
    Z = \frac{1}{2} \Tr_{\calY}(H J(\Phi)) +
    \frac{1}{2} \Tr_{\calY}(H J(\Phi))^{\ast} - \varepsilon\,\I_{\calX},
  \end{equation}
  and define
  \begin{equation}
    Y = H - \I_{\calY} \otimes Z.
  \end{equation}
  Also define  
  \begin{equation}
    A = H - \I_{\calY}\otimes \Tr_{\calY}(H J(\Phi)) - P,
  \end{equation}
  so that $\norm{A}_{\infty} = \varepsilon$, and therefore
  \begin{equation}
    \biggnorm{\frac{1}{2}A + \frac{1}{2}A^{\ast}}_{\infty} \leq \varepsilon.
  \end{equation}
  It is the case that
  \begin{equation}
    Y = P + \frac{A + A^{\ast}}{2} + \varepsilon\,\I_{\calY}\otimes\I_{\calX}
    \geq P,
  \end{equation}
  and therefore $Y\in\Pos(\calY\otimes\calX)$.
  
  Now consider the Lagrange dual problem \eqref{eq:Lagrange-dual-problem}
  associated with the minimization of $f$ over all channels in
  $\C(\calX,\calY)$.
  As $Z$ is Hermitian and $Y$ is positive semidefinite, $(Y,Z)$ is a dual
  feasible solution to this dual problem.
  As $H \in\partial f(J(\Phi))$, one has that
  \begin{equation}
    f(X) - \langle H, X\rangle
    \geq f(J(\Phi)) - \langle H, J(\Phi)\rangle
  \end{equation}
  for every $X\in\Herm(\calY\otimes\calX)$.
  Therefore, when the dual objective function $g$ is evaluated at $(Y,Z)$,
  we find that
  \begin{equation}
    \begin{aligned}
      g(Y,Z) & = \Tr(Z) + \inf_{X\in\Herm(\calY\otimes\calX)}
      \bigl( f(X) - \bigl\langle X, Y + \I_{\calY}\otimes Z \bigr\rangle
      \bigr)\\[1mm]
      & = \langle H, J(\Phi)\rangle - \varepsilon\,\dim(\calX)
      + \inf_{X\in\Herm(\calY\otimes\calX)}
      \bigl( f(X) - \langle H,X\rangle \bigr)\\[1mm]
      & \geq f(J(\Phi)) - \varepsilon\,\dim(\calX).
    \end{aligned}
  \end{equation}
  The theorem follows by weak duality.  
\end{proof}

We note that the previous theorem does not require the function $f$ to be
convex, or for the associated optimization problem to satisfy the conditions
of Slater's theorem.
However, having knowledge of the subdifferential $\partial f(X)$ at an operator
$X$ requires knowledge of the global behavior of the function $f$ when $f$ is
not convex, so the usefulness of Theorem~\ref{thm:approxoptimal} may be limited
when $f$ is not convex.

\begin{remark}
  Theorem~\ref{thm:approxoptimal} implies that every feasible primal solution
  $X$ to the optimization problem in \eqref{eq:cvxchannels-Choi} for which
  $\partial f(X)$ is non-empty provides both an upper bound \emph{and} a lower
  bound to the optimal value. Indeed, from \eqref{eq:epsbound} one has that
  \begin{equation}
    f(X)-\varepsilon \leq \inf_{Y\in J(\C(\calX,\calY))}f(Y) \leq f(X)
  \end{equation}
  provided that  one can compute the value of $\varepsilon$ as defined in
  \eqref{eq:eps} for some $H\in \partial f(X)$.
\end{remark}

\section{Applications}

In this section we apply the optimality conditions given by
Theorem~\ref{thm:cvxchannelschoigeneral} to a few categories of examples,
including the simple case of channel optimization problems having linear
objective functions and three variants of problems involving quantum state
transformations.
In these examples we will make use of various facts concerning differentiation
for functions mapping between spaces of Hermitian operators; a short discussion
of this topic, along with a lemma that is needed for one of the examples, can
be found in Appendix \ref{appendix:gradients}.
Additionally, we remark that Theorem \ref{thm:cvxtransformationfids} can be
used to prove that the generalized fidelity of recovery is multiplicative.

\subsection{Linear objective functions}

We will begin by considering the simple case in which the objective function
$f$ in Theorem~\ref{thm:cvxchannelschoigeneral} is linear.
In this situation, the optimization problem \eqref{eq:cvxchannelschoi}
may be rewritten as
\begin{equation}
  \label{eq:cvxstatediscriminationchannels}
  \begin{aligned}
    &\text{minimize}  &&\langle H,X\rangle\\
    &\text{subject to}&& X\in J(\C(\calX,\calY))
  \end{aligned}
\end{equation}
for some choice of a Hermitian operator $H\in\Herm(\calY\otimes\calX)$.
The subdifferential of the function $f(X) = \langle H,X\rangle$ is given by
$\partial f(X) = \{H\}$ for all $X\in\Herm(\calY\otimes\calX)$.
By Theorem~\ref{thm:cvxchannelschoigeneral} it follows that the Choi operator
$X = J(\Phi)$ of a channel $\Phi\in\C(\calX,\calY)$ is optimal for the
problem \eqref{eq:cvxstatediscriminationchannels} if and only if
\begin{equation}
  \label{eq:linear-criterion}
  \Tr_{\calY}(HX)\in\Herm(\calX)
  \quad\text{and}\quad
  H \geq \I_{\calY}\otimes\Tr_{\calY}(HX).
\end{equation}
We observe that this optimality criterion can alternatively be obtained through
semidefinite programming duality and complementary slackness.
(See, for instance, Exercise 3.5 of \cite{watrous2018}, observing that the
inequality is reversed in that exercise because the optimization problem is
expressed as a maximization rather than a minimization.)

The problem of \emph{minimum error state discrimination}, which was mentioned
in the introduction, is a special case in which the function $f$ in the
optimization problem \eqref{eq:cvxchannelschoi} is linear.
Consider an ensemble of states, which represents the random selection of one of
a finite number of quantum states according to a given probability
distribution.
Formally speaking, an ensemble is described by a collection
$\{\rho_1,\ldots,\rho_n\}\subseteq\D(\calX)$ of density operators together with
a probability vector $p = (p_1,\ldots,p_n)$.
The problem of minimum-error state discrimination seeks a measurement on
the system represented by the space $\calX$ that identifies, with the minimum
possible probability of error, a state chosen randomly according to this
ensemble.

For a given choice of a measurement, represented by operators
$\{P_1,\ldots,P_n\}\subseteq\Pos(\calX)$, the error probability incurred by
this measurement can be expressed as
\begin{equation}
  \label{eq:state-discrimination-error-probabiity}
  \sum_{k = 1}^n \bigl\langle \rho - p_k \rho_k , P_k\bigr\rangle,
\end{equation}
where
\begin{equation}
  \rho = \sum_{k=1}^n p_k \rho_k
\end{equation}
is the average state of the ensemble.
A minimization of the error probability
\eqref{eq:state-discrimination-error-probabiity} over all measurements
$\{P_1,\ldots,P_n\}$ can be represented as an optimization of the form
\eqref{eq:cvxstatediscriminationchannels} by letting $\mathcal{Y}=\CC^n$ and
setting
\begin{equation}
  H = \sum_{k=1}^n E_{k,k}\otimes (\rho - p_k \rho_k)^{\t}.
\end{equation}
The measurement described by $\{P_1,\ldots,P_n\}$, which may alternatively be
represented by a quantum-to-classical channel whose Choi operator is defined as
\begin{equation}
  X = \sum_{k = 1}^n E_{k,k}\otimes P_k^{\t},
\end{equation}
is optimal for the minimization of the error probability 
\eqref{eq:state-discrimination-error-probabiity} if and only of the conditions
\eqref{eq:linear-criterion} hold.
These conditions can be simplified by first calculating that
\begin{equation}
  \Tr_{\calY} (HX) = \sum_{k = 1}^n (\rho - p_k \rho_k)^{\t} P_k^{\t}
  = \rho^{\t} - \sum_{k = 1}^n p_k \rho_k^{\t} P_k^{\t},
\end{equation}
then observing that this operator is Hermitian if and only if the operator
\begin{equation}
  \label{eq:HYKL-1}
  \sum_{k = 1}^n p_k P_k \rho_k
\end{equation}
is Hermitian, and finally noting that the inequality
$H \geq \I_{\calY}\otimes\Tr_{\calY}(HX)$ is equivalent to
\begin{equation}
  \sum_{j = 1}^n E_{j,j} \otimes p_j \rho_j \leq \I_{\calY}\otimes
  \sum_{k = 1}^n p_k P_k \rho_k,
\end{equation}
which may alternatively be expressed as
\begin{equation}
  \label{eq:HYKL-2}
  \sum_{k = 1}^n p_k P_k \rho_k \geq p_j \rho_j
  \quad\text{(for all $j=1,\ldots,n$)}.
\end{equation}
The conditions that the operator \eqref{eq:HYKL-1} is Hermitian and satisfies
the inequalities \eqref{eq:HYKL-2} comprise the Holevo--Yuen--Kennedy--Lax
conditions for measurement optimality
\cite{Holevo73-Kyoto,Holevo73-statistical,YuenKL70,YuenKL75}.

\subsection{State transformation}

The second category of channel optimization problems we consider involves
channels that transform one state into another in such a way that the distance
to a target state is minimized.
One may consider any number of specific measures of distance in such a problem;
we will analyze measures based on the fidelity, trace distance, and quantum
relative entropy.
For each of these measures, we will consider the situation in which two
bipartite states, $\rho\in\D(\calX\otimes\calZ)$ and
$\sigma\in\D(\calY\otimes\calZ)$, for complex Euclidean spaces $\calX$,
$\calY$, and $\calZ$, are given.
The optimization problem to be considered is to minimize the distance (or
maximize the similarity) between the states
$\bigl(\Phi\otimes\I_{\L(\calZ)}\bigr)(\rho)$ and $\sigma$, with respect to the
measure under consideration, over all possible channels
$\Phi\in\C(\calX,\calY)$.

Optimization problems of this sort were suggested in 
\cite{SeshadreesanW15} (see Remark~6) in the special case that
$\calY = \calX\otimes\calW$ and $\rho = \Tr_{\calW}(\sigma)$, so that the
aim of the channel $\Phi$ is to ``recover'' the original state $\sigma$ from
the portion of it represented by $\rho$.
The resulting quantities are known as the \emph{fidelity of recovery},
\emph{relative entropy of recovery}, and so on.
These and related measures have been studied in a number of recent papers
(such as \cite{FawziR15,BrandaoHOS15,CooneyHMOSWW16,BertaT16,BertaFT17}), and
the more general situation in which $\rho$ and $\sigma$ are arbitrary states
arises naturally in this study.
(For example, \cite{BertaT16} uses the term
\emph{generalized fidelity of recovery} in the situation mentioned above in the
case of the fidelity.) 

As a further application of the methods introduced in this paper,
it is straightforward to show (by applying the result of Theorem
\ref{thm:cvxtransformationfids}) that the generalized fidelity of recovery is
multiplicative. This fact has already been shown in \cite{BertaT16}, where it
was shown by examining the fidelity of recovery as a semidefinite program.
A detailed proof using Theorem~\ref{thm:cvxtransformationfids} is presented in
Appendix \ref{app:multiplicativity}.

When analyzing optimality conditions for these problems, it will be helpful to
refer to the \emph{evaluation map} corresponding to the operator
$\rho\in\D(\calX\otimes\calZ)$.
This is the uniquely determined completely positive map
$\Psi_{\rho} \in \CP(\calX,\calZ)$ that satisfies
\begin{equation}
  \bigl( \I_{\L(\calY)}\otimes\Psi_{\rho} \bigr)(J(\Phi))
  = \bigl( \Phi \otimes \I_{\L(\calZ)} \bigr)(\rho)
\end{equation}
for every complex Euclidean space $\calY$ and every channel
$\Phi\in\C(\calX,\calY)$.
The relationship between the map $\Psi_{\rho}$ and the state $\rho$ is closely
related to the Choi representation of maps: assuming for the moment that
$\calX = \CC^n$, one has that
\begin{equation}
  \rho = \sum_{j,k=1}^n E_{j,k} \otimes \Psi_\rho(E_{j,k}).
\end{equation}
That is, up to swapping the tensor factors corresponding to $\calX$ and
$\calZ$, the state $\rho$ is the Choi operator of the map $\Psi_{\rho}$.

\subsubsection*{Objective functions based on the fidelity}

For positive semidefinite operators $P,Q\in\Pos(\calX)$, one defines the
fidelity between $P$ and $Q$ as
\begin{equation}
  \F(P,Q) = \Bignorm{\sqrt{P}\sqrt{Q}}_1 = \Tr \sqrt{\sqrt{P}Q\sqrt{P}}.
\end{equation}
The first variant of the optimal state transformation problem we will consider
is as follows:
\begin{equation}
  \label{eq:cvxtransformationpartialaccess}
  \begin{aligned}
    &\text{minimize}  &&
    -\F\bigl(\sigma,\bigl(\Phi\otimes\I_{\L(\calZ)}\bigr)(\rho)\bigr)\\
    & \text{subject to} && \Phi\in\C(\calX,\calY).
  \end{aligned}
\end{equation}
The fidelity function is jointly concave (see for example Corollary 3.26 in
\cite{watrous2018}), and therefore is concave in each of its arguments, from
which it follows that this problem is a convex optimization problem.
It is possible to express this optimization problem as a semidefinite program,
as is demonstrated in \cite{BertaFT17}.

The following theorem establishes optimality conditions for a channel
$\Phi\in\C(\calX,\calY)$ in the optimization problem 
\eqref{eq:cvxtransformationpartialaccess} under the assumption that the
operator $\Tr_{\calX}(\rho)$ is positive definite.
We note that the theorem statement does not actually require $\rho$ and
$\sigma$ to have unit trace---they can be arbitrary positive semidefinite
operators, but we nevertheless use the letters $\rho$ and $\sigma$ to make the
connection to the optimization problem
\eqref{eq:cvxtransformationpartialaccess} clear.

\begin{theorem}
  \label{thm:cvxtransformationfids}
  Let $\rho\in\Pos(\calX\otimes\calZ)$ and $\sigma\in\Pos(\calY\otimes\calZ)$
  be positive semidefinite operators, for complex Euclidean spaces $\calX$,
  $\calY$, and $\calZ$, and assume that $\Tr_{\calX}(\rho)$ is a positive
  definite operator.
  A channel $\Phi\in\C(\calX,\calY)$ is optimal for the optimization problem 
  \eqref{eq:cvxtransformationpartialaccess} if and only if the following two
  conditions are met:
  \begin{enumerate}
  \item[1.]
    $\im(\sigma) \subseteq
    \im\bigl(\bigl(\Phi\otimes\I_{\L(\calZ)}\bigr)(\rho) \bigr)$.
  \item[2.]
    The operator
    \begin{equation}
      \label{eq:fidelity-H}
      H =-\frac{1}{2}(\I_{\L(\calY)}\otimes \Psi_{\rho}^{\ast})
      \Bigl(
      \sqrt{\sigma}\tinyspace
      \bigl( \sqrt{\sigma}\, (\Phi\otimes\I_{\L(\calZ)})(\rho)
      \sqrt{\sigma}\bigr)^{-\frac{1}{2}}
      \sqrt{\sigma}
      \Bigr)
    \end{equation}
    satisfies
    \begin{equation}
      \Tr_{\calY}(H J(\Phi)) \in \Herm(\calX)
      \quad\text{and}\quad
      H \geq \I_{\calY}\otimes\Tr_{\calY}(H J(\Phi)).
    \end{equation}
    (As per the convention mentioned in Section~\ref{sec:preliminaries},
    the inverse in \eqref{eq:fidelity-H} refers to the Moore--Penrose
    pseudo-inverse in case $\sigma$ does not have full rank.)
  \end{enumerate}
\end{theorem}

\begin{remark}
  \label{remark:fidelity-transformation}
  The theorem assumes that $\Tr_{\calX}(\rho)$ has full rank, as this
  assumption allows for a cleaner theorem statement.
  It is, however, straightforward to apply the theorem to a situation in which
  $\Tr_{\calX}(\rho)$ does not have full rank.
  Specifically, for an arbitrary choice of $\rho$ and $\sigma$, one may take
  $B\in\L(\calV,\calZ)$ to be an isometry for which $B B^{\ast}$ is the
  projection onto the image of $\Tr_{\calX}(\rho)$, and then observe that by
  replacing $\rho$ and $\sigma$ with
  $(\I_{\calX}\otimes B^{\ast})\rho(\I_{\calX}\otimes B)$ and
  $(\I_{\calY}\otimes B^{\ast})\sigma(\I_{\calY}\otimes B)$, respectively,
  an equivalent problem is obtained that satisfies the assumptions of the
  theorem.
\end{remark}

\begin{remark}\label{rem:geomean}
The argument of the expression in \eqref{eq:fidelity-H} has the form of an
\emph{operator geometric mean}, which is defined as follows. For positive
definite operators $P,Q\in\Pd(\calX)$ on a complex Euclidean space $\calX$,
their geometric mean is the operator $P\,\#\,Q\in\Pd(\calX)$ defined as
\[
P\,\#\,Q = \sqrt{P}
\Bigl( \sqrt{P^{-1}} Q \sqrt{P^{-1}} \Bigr)^{\frac{1}{2}}\sqrt{P}.
\]
If one takes $A\in\L(\calW,\calY\otimes\calZ)$ to be an isometry for which $A
A^{\ast} = \Pi_{\im(\sigma)}$, the operator $H$ in~\eqref{eq:fidelity-H} may be
expressed as
\begin{equation}
  H =-\frac{1}{2}(\I_{\L(\calY)}\otimes \Psi_{\rho}^{\ast})
  \Bigl(A\Bigl(
  (A^*\sigma A)\,\#\,(A^*YA)^{-1}
  \Bigr)A^*\Bigr)
\end{equation}
where 
$Y=(\Phi\otimes\I_{\L(\calZ)})(\rho)$.
If the operator $\sigma$ is positive definite (in which case
$A=\I_{\calY\otimes\calZ}$), this expression simplifies to
\begin{equation}
  H = -\frac{1}{2}(\I_{\L(\calY)}\otimes
  \Psi_{\rho}^{\ast})\bigl(\sigma\,\#\,Y^{-1}\bigr).
\end{equation}
For further discussion of the operator geometric mean see Section 4.1 in
\cite{Bhatia2007}.
\end{remark}

\begin{proof}[Proof of Theorem~\ref{thm:cvxtransformationfids}]
  Let $r = \op{rank}(\sigma)$, let $\calW = \complex^r$, and let
  $A\in\L(\calW,\calY\otimes\calZ)$ be any isometry for which
  $A A^{\ast} = \Pi_{\im(\sigma)}$ (the projection onto the image of $\sigma$).
  Define a function
  \begin{equation}
    g:\Herm(\calY\otimes\calZ)\rightarrow\RR\cup\{\infty\}
  \end{equation}
  as
  \begin{equation}
    g(Y) = \begin{cases}
      -\F(A^{\ast}\sigma A, A^{\ast} Y A ) &
      \text{if}\;A^{\ast} Y A \in\Pos(\calW)\\
      \infty & \text{otherwise},
    \end{cases}
  \end{equation}
  for all $Y\in\Herm(\calY\otimes\calZ)$, and observe that
  $g(Y) = -\F(\sigma,Y)$ for every $Y\in\Pos(\calY\otimes\calZ)$.
  
  For a given operator $Y\in\Herm(\calY\otimes\calZ)$ satisfying
  $A^{\ast} Y A \in\Pos(\calW)$, there are two cases for the subdifferential
  $\partial g(Y)$.
  \begin{trivlist}
  \item Case 1: $A^{\ast} Y A$ is positive definite.
    In this case $g$ is differentiable at $Y$, and
    \begin{equation}
      \nabla g(Y) = -\frac{1}{2} A\left((A^*\sigma A)\,\#\, (A^*YA)^{-1}\right)A^*
    \end{equation}
    which follows from Lemma~\ref{lemma:fidelity-derivative} (as stated and proved
    in Appendix \ref{appendix:gradients}).
    It therefore follows that
    \begin{equation}
      \partial g(Y) = \biggl\{-\frac{1}{2} A\left((A^*\sigma A)\,\#\, (A^*YA)^{-1}\right)A^*
      \biggr\}.
    \end{equation}
  \item
    Case 2: $A^{\ast} Y A$ is not positive definite.
    In this case, $\partial g(Y) = \varnothing$, which also follows from
    Lemma~\ref{lemma:fidelity-derivative}.
  \end{trivlist}

  Next, define
  \begin{equation}
    \Lambda = \I_{\L(\calY)} \otimes \Psi_{\rho},
  \end{equation}
  and observe that
  \begin{equation}
    (g\circ\Lambda)(J(\Phi)) = 
    -\F\bigl(\sigma,\bigl(\Phi\otimes\I_{\L(\calZ)}\bigr)(\rho)\bigr)
  \end{equation}
  for every channel $\Phi\in\C(\calX,\calY)$.
  It is the case that
  $\Lambda(\I_{\calY}\otimes\I_{\calX})= \I_{\calY}\otimes\Tr_{\calX}(\rho)$,
  which is positive definite by assumption.
  As $\Pd(\calY\otimes\calZ)\subset\relint(\dom(g))$, it follows that
  \begin{equation}
   \im(\Lambda) \cap \relint(\dom(g)) \neq \varnothing,
  \end{equation}
  and therefore, by Proposition~\ref{prop:subdifferential-chain-rule},
  \begin{equation}
    \partial (g\circ\Lambda)(X) = \Lambda^{\ast}(\partial g(\Lambda(X)))
  \end{equation}
  for every $X\in\Herm(\calY\otimes\calX)$.

  The theorem now follows from Theorem~\ref{thm:cvxchannelschoigeneral}.
  In greater detail, if $\Phi$ is optimal for the problem 
  \eqref{eq:cvxtransformationpartialaccess}, there must
  exist an operator $H\in\partial (g\circ\Lambda)(J(\Phi))$ such that
  \begin{equation}
    \Tr_{\calY}(H J(\Phi))\in\Herm(\calX)
    \quad\text{and}\quad
    H \geq \I_{\calY}\otimes\Tr_{\calY}(H J(\Phi)).
  \end{equation}
  Case 1 described above must therefore hold when $Y = \Lambda(J(\Phi))$, for
  otherwise the subdifferential $\partial (g\circ\Lambda)(J(\Phi))$ would be
  empty.
  It follows that
  \begin{equation}
    H =-\frac{1}{2}(\I_{\L(\calY)}\otimes \Psi_{\rho}^{\ast})
    \bigl(A\bigl(
    (A^*\sigma A)\,\#\,(A^*(\Phi\otimes\I_{\L(\calZ)})(\rho)A)^{-1}
    \bigr)A^*\bigr)
  \end{equation}
  and straightforward manipulation reveals that 
  \begin{equation}
    H = -\frac{1}{2}(\I_{\L(\calY)}\otimes \Psi_{\rho}^{\ast})
    \Bigl(
    \sqrt{\sigma}\tinyspace
    \bigl( \sqrt{\sigma}\, (\Phi\otimes\I_{\L(\calZ)})(\rho)
    \sqrt{\sigma}\bigr)^{-\frac{1}{2}}
    \sqrt{\sigma}
    \Bigr)
  \end{equation}
  (where, as always, we interpret the inverse as the Moore--Penrose
  pseudo-inverse in the case when the given operators are not positive
  definite).
  The second condition in the statement of the theorem now follows.

  Conversely, if the two conditions in the statement of the theorem hold, it
  follows that $H\in\partial (g\circ\Lambda)(J(\Phi))$, and moreover that
  \begin{equation}
    \Tr_{\calY}(H J(\Phi))\in\Herm(\calX)
    \quad\text{and}\quad
    H \geq \I_{\calY}\otimes\Tr_{\calY}(H J(\Phi)).
  \end{equation}
  One concludes by Theorem~\ref{thm:cvxchannelschoigeneral} that
  $\Phi$ is optimal for the optimization problem 
  \eqref{eq:cvxtransformationpartialaccess}.  
\end{proof}

An interesting special case of the optimization problem
\eqref{eq:cvxtransformationpartialaccess} is when the states $\rho$ and
$\sigma$ take the form
\begin{equation}
  \rho = \sum_{k = 1}^n p_k \rho_k \otimes E_{k,k}
  \quad\text{and}\quad
  \sigma = \sum_{k = 1}^n p_k \sigma_k \otimes E_{k,k}
\end{equation}
for a probability vector $(p_1,\ldots,p_n)$ and collections of states
\begin{equation}
  \{\rho_1,\ldots,\rho_n\}\subseteq\D(\calX)
  \quad\text{and}\quad
  \{\sigma_1,\ldots,\sigma_n\}\subseteq\D(\calY).
\end{equation}
The objective function simplifies in this case to
\begin{equation}
  -\sum_{k=1}^n p_k \F(\sigma_k,\Phi(\rho_k)),
\end{equation}
and therefore the optimization concerns the average fidelity with which
a channel $\Phi$ maps each $\rho_k$ to $\sigma_k$.
The operator
\begin{equation}
  H = -\frac{1}{2} \bigl(\I_{\L(\calY)} \otimes \Psi_{\rho}^{\ast}\bigr)
  \Bigl( \sqrt{\sigma} \Bigl( \sqrt{\sigma}
  \bigl( \Phi \otimes \I_{\L(\calZ)} \bigr)(\rho)
  \sqrt{\sigma}\Bigr)^{-\frac{1}{2}}\sqrt{\sigma}\Bigr)
\end{equation}
simplifies in this case to
\begin{equation}
  H = -\frac{1}{2} \sum_{k = 1}^n
  p_k \sqrt{\sigma_k}\Bigl(\sqrt{\sigma_k} \tinyspace\Phi(\rho_k)
  \sqrt{\sigma_k}\Bigr)^{-\frac{1}{2}}\sqrt{\sigma_k} \otimes \rho_k^{\t}.
\end{equation}

An analogous criterion for optimality can also be found when the figure of
merit is the square of the fidelity, which is also jointly concave (see for
example Property 9.2.2 in \cite{Wilde17}).
The corresponding optimization problem is the following: 
\begin{equation}\label{eq:cvxtransformationfids2}
  \begin{aligned}
    &\text{minimize}  &&
    -\sum_{k=1}^m p_{k} \F\bigl(\sigma_{k},\Phi(\rho_{k})\bigr)^2\\
    &\text{subject to}&& \Phi\in\C(\calX,\calY).
  \end{aligned}
\end{equation}
By differentiating the objective function in \eqref{eq:cvxtransformationfids2},
one finds that a channel $\Phi$ is optimal precisely when the same conditions
in Theorem \ref{thm:cvxtransformationfids} are met, but with
$H\in\Herm(\calY\otimes\calX)$ given by
\begin{equation}
  H = - \sum_{k=1}^n p_k
  \F\bigl(\sigma_k,\Phi(\rho_k)\bigr) \sqrt{\sigma}
  \Bigl( \sqrt{\sigma} \bigl(\Phi\otimes\I_{\L(\calZ)}\bigr)(\rho)
  \sqrt{\sigma}\Bigr)^{-\frac{1}{2}}\sqrt{\sigma}
  \otimes \rho_k^{\t}.
\end{equation}

\subsubsection*{Objective functions based on trace distance}

Next we consider an optimization problem that is analogous to
\eqref{eq:cvxtransformationpartialaccess}, but based on the trace distance
rather than the fidelity:
\begin{equation}
  \label{eq:cvxtransformationtrdist}
  \begin{aligned}
    &\text{minimize} && \bignorm{\sigma-(\Phi\otimes\I_{\L(\calZ)})(\rho)}_1 \\
    &\text{subject to}&& \Phi\in\C(\calX,\calY).
  \end{aligned}
\end{equation}
The objective function of this optimization problem is convex (with respect to
$\Phi$), so we may use Theorem \ref{thm:cvxchannelschoigeneral} to obtain
optimality conditions for this problem.
We note that, similar to its variant based on the fidelity described above,
this optimization problem can be represented as a semidefinite program.

For this optimization problem, the optimality conditions we obtain
may not be efficiently checkable.
We do, however, obtain an efficiently checkable condition that is sufficient
for optimality, and we conjecture that this condition is also a necessary for
optimality.

\begin{corollary}\label{thm:cvxtransformationtrdist}
  Let $\rho\in\Pos(\calX\otimes\calZ)$ and $\sigma\in\Pos(\calY\otimes\calZ)$
  be positive semidefinite operators, for complex Euclidean spaces $\calX$,
  $\calY$, and $\calZ$, and let $\Phi\in\C(\calX,\calY)$ be a channel.
  The channel $\Phi$ is optimal for the state transformation problem in
  \eqref{eq:cvxtransformationtrdist} if and only if there exists an operator
  $Y\in\Herm(\calY\otimes\calZ)$ with $\norm{Y}_{\infty} = 1$ such that the
  following conditions are satisfied:
  \begin{enumerate}
  \item
    It is the case that
    $\bignorm{\sigma - (\Phi\otimes\I_{\L(\calZ)})(\rho)}_1
    = \bigl\langle Y ,
    \sigma - (\Phi\otimes\I_{\L(\calZ)})(\rho)\bigr\rangle$.
  \item
    The operator
    \begin{equation}
      \label{eq:Htracedist}
      H = \bigl(\I_{\L(\calY)}\otimes \Psi_{\rho}^{\ast}\bigr)(Y)
    \end{equation}
    satisfies
    \begin{equation}
      \Tr_{\calY}(HJ(\Phi))\in\Herm(\calX) \quad\text{and}\quad
      H \geq \I_{\calY}\otimes \Tr_{\calY}\bigl(HJ(\Phi)\bigr).
    \end{equation}
  \end{enumerate}
\end{corollary}

Corollary \ref{thm:cvxtransformationtrdist} follows directly from
Theorem \ref{thm:cvxchannelschoigeneral} and applying the rules of
subdifferentiation presented in Section \ref{sec:cvx}.
Indeed, operators of the form in \eqref{eq:Htracedist} are precisely the
elements of the subdifferential of the objective function in
\eqref{eq:cvxtransformationtrdist} at $J(\Phi)$.
As a generalization, one can replace the trace norm $\lVert\cdot\rVert_1$ in
the statement of Corollary \ref{thm:cvxtransformationtrdist} with any other norm
on operators, and replace $\lVert\cdot\rVert_\infty$ with the corresponding
dual norm.

In the event that the operator $\sigma - (\Phi\otimes\I_{\L(\calZ)})(\rho)$
arising in Corollary~\ref{thm:cvxtransformationtrdist} has no zero eigenvalues,
there is a unique choice of $Y$ for which the first condition of the theorem
holds.
Specifically, if
\begin{equation}
  \label{eq:spectral-decomposition-of-difference}
  \sigma - (\Phi\otimes\I_{\L(\calZ)})(\rho) = \sum_{k = 1}^m
  \lambda_k \Pi_k
\end{equation}
is a spectral decomposition where each $\lambda_k$ is nonzero, then the unique
operator $Y$ satisfying condition 1 in
Corollary~\ref{thm:cvxtransformationtrdist} is given by
\begin{equation}
  \label{eq:expression-of-Y}
  Y = \sum_{k = 1}^m \op{sign}(\lambda_k)\, \Pi_k.
\end{equation}
In this case it is sufficient for the second condition to be checked for this
unique choice of~$Y$, yielding an efficiently checkable optimality criterion.
However, if it is the case that
$\sigma - (\Phi\otimes\I_{\L(\calZ)})(\rho)$ has one or more zero eigenvalues,
then the first condition holds for a continuum of choices of $Y$, and from
Corollary \ref{thm:cvxtransformationtrdist} we conclude only that the optimality
of $\Phi$ is equivalent to the existence of at least one such choice of $Y$ for
which the second statement in the theorem hold.

It is reasonable, though, to view the operator $Y$ defined by
\eqref{eq:expression-of-Y}, where now it is to be understood that
$\op{sign}(0) = 0$, as a natural selection of an operator through which
optimality may be verified.
We conjecture, based on numerical evidence, that this choice yields an
efficiently checkable necessary and sufficient optimality condition.

\begin{conjecture}
  Let $\rho\in\D(\calX\otimes\calZ)$ and $\sigma\in\D(\calY\otimes\calZ)$
  be density operators, for complex Euclidean spaces $\calX$, $\calY$, and
  $\calZ$, and let $\Phi\in\C(\calX,\calY)$ be a channel.
  Let
  \begin{equation}
    \sigma - (\Phi\otimes\I_{\L(\calZ)})(\rho) = \sum_{k = 1}^m
    \lambda_k \Pi_k
  \end{equation}
  be a spectral decomposition, and define
  \begin{equation}
    Y = \sum_{k = 1}^m \op{sign}(\lambda_k)\, \Pi_k,
  \end{equation}
  where $\op{sign}(\alpha) = 1$ and $\op{sign}(-\alpha) = -1$ for
  all $\alpha > 0$ and $\op{sign}(0) = 0$.  
  The channel $\Phi$ is optimal for the state transformation problem in
  \eqref{eq:cvxtransformationtrdist} if and only if the operator
  \begin{equation}
    H = \bigl(\I_{\L(\calY)}\otimes \Psi_{\rho}^{\ast}\bigr)(Y)
  \end{equation}
  satisfies
  \begin{equation}
    \Tr_{\calY}(HJ(\Phi))\in\Herm(\calX) \quad\text{and}\quad
    H \geq \I_{\calY}\otimes \Tr_{\calY}\bigl(HJ(\Phi)\bigr).
  \end{equation}
\end{conjecture}

\subsubsection*{Objective functions based on relative entropy}

Finally, we consider a variant of the optimal state transformation problem
based on the quantum relative entropy.
For positive semidefinite operators $P,Q,\in\Pos(\calX)$, the quantum relative
entropy of $P$ with respect to $Q$ is defined as
\begin{equation*}
  \D(P\lVert Q) = \begin{cases}
    \Tr(P\log(P)) - \Tr(P\log(Q)) & \text{if}\;\im(P)\subseteq \im(Q)\\
    \infty & \text{otherwise}.
  \end{cases}
\end{equation*}
The specific variant of the problem to be considered is
\begin{equation}
  \label{eq:cvxtransformation-relative-entropy-1}
  \begin{aligned}
    &\text{minimize} && \D\bigl(\sigma
    \tinyspace\big\|\tinyspace
    (\Phi\otimes\I_{\L(\calZ)})(\rho)\bigr)\\
    &\text{subject to} && \Phi\in\C(\calX,\calY).
  \end{aligned}
\end{equation}
The relative entropy is jointly convex, which implies that it is convex in its
second argument, and therefore the problem above is a convex optimization
problem.
Because the relative entropy can be approximated through the use of
semidefinite programming \cite{Fawzi2018,Fawzi2017a}, it is possible to
efficiently approximate the optimization problem
\eqref{eq:cvxtransformation-relative-entropy-1} on a computer.

\begin{theorem}
  \label{thm:cvxtransformationrelent}
  Let $\rho\in\Pos(\calX\otimes\calZ)$ and $\sigma\in\Pos(\calY\otimes\calZ)$
  be positive semidefinite operators, for complex Euclidean spaces $\calX$,
  $\calY$, and $\calZ$, and assume that $\Tr_{\calX}(\rho)$ is a positive
  definite operator.
  A channel $\Phi\in\C(\calX,\calY)$ is optimal for the optimization problem 
  \eqref{eq:cvxtransformation-relative-entropy-1} if and only if the following
  two conditions are met:
  \begin{enumerate}
  \item[1.]
    $\im(\sigma) \subseteq \im\bigl( (\Phi\otimes\I_{\L(\calZ)})(\rho) \bigr)$.
  \item[2.]
    The operator
    \begin{equation}
      H =
      -(\I_{\L(\calY)} \otimes \Psi_{\rho}^{\ast})
      \bigl(
      D\log\bigl( \Pi \bigl(\Phi\otimes\I_{\L(\calZ)}\bigr)(\rho)
      \Pi\bigr)(\sigma) \bigr)
    \end{equation}
    satisfies
    \begin{equation}
      \Tr_{\calY}(H J(\Phi)) \in \Herm(\calX)
      \quad\text{and}\quad
      H \geq \I_{\calY}\otimes\Tr_{\calY}(H J(\Phi)).
    \end{equation}
    Here, $\Pi$ denotes the projection onto the image of $\sigma$ and
    $D\log(P)$ denotes the differential operator of the
    logarithm function at the operator $P$ (as described in \eqref{eq:DlogYZ}
    at the end of Appendix~\ref{appendix:gradients}).
  \end{enumerate}
\end{theorem}

\begin{proof}
  If the first condition does not hold for a given channel $\Phi$, then the
  objective function in \eqref{eq:cvxtransformation-relative-entropy-1} takes
  an infinite value.
  However, by the assumption that $\Tr_{\calX}(\rho)$ is positive definite,
  one has that the channel
  \begin{equation}
    \Omega(X) = \frac{\Tr(X) \I_{\calY}}{\dim(\calY)}
  \end{equation}
  yields a finite value for the same objective function, implying that
  $\Phi$ is not optimal.
  If $\Phi$ is optimal, the first condition must therefore hold.
  It remains to prove that if $\Phi$ satisfies the first condition, then
  $\Phi$ is optimal if and only if the second condition holds.

  Let $r$ be the rank of $\sigma$, let $\calW = \CC^r$, and let
  $A\in\L(\calW,\calY\otimes\calZ)$ be an isometry that satisfies
  $A A^{\ast} = \Pi_{\im(\sigma)}$.
  Define a linear map $\Xi: \Herm(\calY\otimes\calX) \rightarrow \Herm(\calW)$
  as
  \begin{equation}
    \Xi(X) = A^{\ast} \bigl(\I_{\L(\calY)}\otimes \Psi_{\rho}\bigr)(X) A
  \end{equation}
  for all $X\in\Herm(\calY\otimes\calX)$.
  For every channel $\Phi\in\C(\calX,\calY)$ it is the case that
  \begin{equation}
    \D\bigl(\sigma
    \tinyspace\big\|\tinyspace
    (\Phi\otimes\I_{\L(\calZ)})(\rho)\bigr)
    = \D\bigl(A^{\ast} \sigma A \tinyspace\big\|\tinyspace \Xi(J(\Phi))\bigr).
  \end{equation}
  With this observation in mind, define a function
  $f:\Herm(\calY\otimes\calX)\rightarrow\RR\cup\{\infty\}$ as
  \begin{equation}
    f(X) = \begin{cases}
      \D\bigl(A^{\ast} \sigma A \tinyspace\big\|\tinyspace \Xi(X)\bigr)
      & \text{if}\; \Xi(X) \in \Pd(\calW)\\
      \infty & \text{otherwise,}
    \end{cases}
  \end{equation}
  so that a given channel $\Phi \in \C(\calX,\calY)$ is optimal for the
  problem \eqref{eq:cvxtransformation-relative-entropy-1} if and only if
  $J(\Phi)$ is optimal for the problem
  \begin{equation}
    \begin{aligned}
      &\text{minimize} && f(X)\\
      &\text{subject to} && X\in J(\C(\calX,\calY)).
    \end{aligned}
  \end{equation}

  The function $f$ is differentiable at every operator
  $X\in\Herm(\calY\otimes\calX)$ in its domain, with its gradient being
  \begin{equation}\begin{split}
      \nabla f(X) &= 
      -\Xi^{\ast} \bigl(
      D\log (\Xi(X))(A^{\ast} \sigma A)\bigr)\\
      &= -\bigl(\I_{\L(\calY)}\otimes\Psi_{\rho}^{\ast}\bigr)
      \bigl(
      D\log \bigl(
      \Pi \bigl(\I_{\L(\calY)}\otimes\Psi_{\rho}\bigr)(X)
      \Pi\bigr)(\sigma)\bigr).
  \end{split}\end{equation}
  As $f$ is differentiable at every point in its domain, it follows that
  \begin{equation}
    \partial f(J(\Phi)) = \{\nabla f(J(\Phi))\}
  \end{equation}
  for every $\Phi\in\C(\calX,\calY)$ for which 
  $\im(\sigma) \subseteq \im\bigl( (\Phi\otimes\I_{\L(\calZ)})(\rho) \bigr)$.
  For a given channel $\Phi\in\C(\calX,\calY)$ for which 
  $\im(\sigma) \subseteq \im\bigl( (\Phi\otimes\I_{\L(\calZ)})(\rho) \bigr)$,
  it therefore follows from Theorem~\ref{thm:cvxchannelschoigeneral} that
  $\Phi$ is optimal if and only if the operator
  \begin{equation}
    H = -\bigl(\I_{\L(\calY)}\otimes\Psi_{\rho}^{\ast}\bigr)
    \bigl(
    D\log\bigl(
    \Pi \bigl(\Phi \otimes \I_{\L(\calZ)}\bigr)(\rho)\Pi\bigr)(\sigma)\bigr)
  \end{equation}
  satisfies
  \begin{equation}
    \Tr_{\calY}(H J(\Phi)) \in \Herm(\calX)
    \quad\text{and}\quad
    H \geq \I_{\calY}\otimes\Tr_{\calY}(H J(\Phi)),
  \end{equation}
  which is the second condition in the statement of the theorem.
\end{proof}

\section*{Acknowledgments}

We thank Jamie Sikora for helpful comments and discussions and Mark Wilde
for valuable suggestions.
This research was supported by Canada's NSERC.


\appendix

\section{Gradients and subdifferentials of functions on matrices}
\label{appendix:gradients}

Results in the main body of this have required the computation of gradients
and subdifferentials for various functions mapping Hermitian operators to the
real numbers.
In this appendix we provide details on these computations.

\subsubsection*{Definitions and basic results}

It is appropriate to begin with some basic definitions.
Throughout this discussion, $\calX$, $\calY$, and $\calZ$ are arbitrary complex
Euclidean spaces.

Suppose that
\begin{equation}
  f:\Herm(\calX) \rightarrow \Herm(\calY)
\end{equation}
is a partial function, meaning that it may only be defined on some subset of
inputs $X\in\Herm(\calX)$.
The function $f$ is \emph{(Fr\'echet) differentiable} at $X\in\Herm(\calX)$
if there exists a linear map
\begin{equation}
  \Phi:\Herm(\calX)\rightarrow\Herm(\calY)
\end{equation}
for which the equation
\begin{equation}
  \lim_{Z\rightarrow 0} \frac{\bignorm{f(X+Z) - f(X) - \Phi(Z)}}{\norm{Z}} = 0
\end{equation}
is satisfied.
If there does exist such a map, it must be unique, and we denote it by
$Df(X)$.
Whenever $f$ is differentiable at $X$, it must be the case that
\begin{equation}
  Df(X)(Z) = \frac{\textup{d}}{\textup{d}t}f(X + tZ) \:\Big|_{t=0}
\end{equation}
for all choices of $Z\in\Herm(\calX)$.
In the special case that $\calY = \CC$, which is equivalent to $f$ taking the
form
\begin{equation}
  f:\Herm(\calX) \rightarrow \RR,
\end{equation}
one has that
\begin{equation}
  Df(X)(Z) = \big\langle \nabla f(X),Z \big\rangle
\end{equation}
for all $Z\in\Herm(\calX)$, assuming $f$ is differentiable at $X$.

The \emph{chain rule} for differentiation states that if
\begin{equation}
  \label{eq:f-and-g}
  f:\Herm(\calX)\rightarrow\Herm(\calY)
  \quad\text{and}\quad
  g:\Herm(\calY)\rightarrow\Herm(\calZ),
\end{equation}
$f$ is differentiable at $X$, and $g$ is differentiable at $Y = f(X)$, then
\begin{equation}
  D(g\circ f)(X)(Z) = Dg(f(X))(Df(X)(Z))
\end{equation}
for all $Z\in\Herm(\calX)$.

\subsubsection*{Affine linear functions}

A function
\begin{equation}
  f:\Herm(\calX) \rightarrow \Herm(\calY)
\end{equation}
is \emph{affine linear} if there exists a linear map
\begin{equation}
  \Phi:\Herm(\calX) \rightarrow \Herm(\calY)
\end{equation}
and an operator $Y \in \Herm(\calY)$ such that
\begin{equation}
  f(X) = \Phi(X) + Y
\end{equation}
for all $X\in\Herm(\calX)$.
Every such function is differentiable at every $X\in\Herm(\calX)$, with its
derivative given by
\begin{equation}
  Df(X) = \Phi.
\end{equation}

For an arbitrary function
\begin{equation}
  g:\Herm(\calY)\rightarrow\Herm(\calZ),
\end{equation}
one therefore finds that
\begin{equation}
  D(g\circ f)(X)(Z) = Dg(f(X))(\Phi(Z)),
\end{equation}
provided that $g$ is differentiable at $f(X)$, and if $g$ takes the form
\begin{equation}
  g:\Herm(\calY)\rightarrow\RR,
\end{equation}
then it is the case that
\begin{equation}
  \nabla(g\circ f)(X) = \Phi^{\ast}(\nabla g(f(X))).
\end{equation}

\subsubsection*{Real-valued functions extended to Hermitian operators}

If $f:\RR \rightarrow \RR$ is a function, then it may be extended to a function
of the form
\begin{equation}
  g:\Herm(\calX) \rightarrow \Herm(\calX)
\end{equation}
in a standard way:
for any choice of $X\in\Herm(\calX)$, one considers the spectral
decomposition
\begin{equation}
  \label{eq:spectral-decomposition-of-X}
  X = \sum_{k = 1}^m \lambda_k \Pi_k
\end{equation}
of $X$, then defines
\begin{equation}
  g(X) = \sum_{k = 1}^m f(\lambda_k) \Pi_k.
\end{equation}
(It is typical that this extended function is given the same name as the
original function on the real numbers, but for the sake of clarity we have
introduced a distinct name for the extended function.)
Naturally, if $f$ is defined only on a subset of $\RR$, then $g$ is defined for
all $X$ whose eigenvalues are contained in the domain of $f$.

The function $g$ is differentiable at every Hermitian operator whose
eigenvalues correspond to differentiable points of the function $f$.
The derivative of $g$ can be described explicitly by first defining a function
\begin{equation}
  h(\alpha,\beta) =
  \begin{cases}
    \frac{f(\alpha) - f(\beta)}{\alpha - \beta} &
    \text{if}\;\alpha\not=\beta\\
    f'(\alpha) & \text{if}\; \alpha = \beta
  \end{cases}
\end{equation}
for every pair of points $(\alpha,\beta)$ for which $f$ is differentiable at
both $\alpha$ and $\beta$.
(The function $h$ is sometimes called the \emph{first divided difference} of
$f$, although this terminology is sometimes limited to the case that
$\alpha\not=\beta$.)
In terms of this function, the derivative of $g$ at an operator $X$ having a spectral decomposition \eqref{eq:spectral-decomposition-of-X} is
\begin{equation}
  \label{eq:derivative-from-divided-difference}
  Dg(X)(Z) = \sum_{j,k=1}^m h(\lambda_j,\lambda_k)
  \Pi_j Z \Pi_k
\end{equation}
for every $Z\in\Herm(\calX)$ (see for example Theorem V.3.3 in \cite{Bhatia1997} and Theorem 3.25 in \cite{Hiai2014}).

\subsubsection*{Gradients of functions involving the fidelity}

Let $\calY$ be a complex Euclidean space, and consider the function
$f:\Herm(\calY)\rightarrow\RR\cup\{\infty\}$ defined as
\begin{equation}
  \label{eq:f-for-fidelity}
  f(Y) = \begin{cases}
    -\Tr \sqrt{Y} & \text{if}\; Y\in\Pos(\calY)\\
    \infty & \text{otherwise.}
  \end{cases}
\end{equation}
This function is differentiable at every positive definite operator
$Y\in\Pd(\calY)$, with its gradient being
\begin{equation}
  \label{eq:gradient-trace-root}
  \nabla f (Y) = -\frac{1}{2}Y^{-\frac{1}{2}}.
\end{equation}
One way to verify this expression is to first consider the function
$g(Y) = \sqrt{Y}$, defined for every positive semidefinite operator
$Y\in\Pos(\calY)$, and to use the formula
\eqref{eq:derivative-from-divided-difference} to conclude that
\begin{equation}
  Dg(Y)(Z) = \sum_{j,k=1}^m \frac{\Pi_j Z \Pi_k}{\sqrt{\lambda_j} +
    \sqrt{\lambda_k}},  
\end{equation}
provided that $Y$ is positive definite and has spectral decomposition
\begin{equation}
  Y = \sum_{k = 1}^m \lambda_k \Pi_k
\end{equation}
(see also Example 3.26 in \cite{Hiai2014}).
The equation \eqref{eq:gradient-trace-root} follows from the chain rule.

\begin{lemma}
  \label{lemma:fidelity-derivative}
  Let $\calX$ be a complex Euclidean space, let $P\in\Pd(\calX)$ be a positive
  definite operator, and define a function
  $g:\Herm(\calX)\rightarrow\RR\cup\{\infty\}$ as
  \begin{equation}
    g(X) = \begin{cases}
      -\F(P,X) & \text{if}\; X\in\Pos(\calX)\\
      \infty & \text{otherwise}.
    \end{cases}
  \end{equation}
  For every $X\in\Pd(\calX)$, the function $g$ is differentiable at $X$, and
  \begin{equation}\label{eq:gradgeometricmean}
    \nabla g(X) =  -\frac{1}{2} \bigl(P\,\#\,X^{-1}\bigr),
  \end{equation}
  where $P\,\#\,X^{-1}$ denotes the operator geometric mean of $P$ and
  $X^{-1}$, as discussed in Remark \ref{rem:geomean}.
  For every operator $X\in\Pos(\calX)$ that is not positive definite, it is the
  case that $\partial g(X) = \varnothing$.  
\end{lemma}

\begin{proof}
  Define a linear map $\Lambda:\Herm(\calX)\rightarrow\Herm(\calY)$ as
  \begin{equation}
    \Lambda(X) = \sqrt{P} X \sqrt{P}
  \end{equation}
  for all $X\in\Herm(\calX)$.
  It is the case that $g = f\circ \Lambda$, where $f$ is as defined in
  \eqref{eq:f-for-fidelity} above.
  By the chain rule for differentiation, one has
  \begin{equation}
    \nabla g(X) = -\frac{1}{2}\Lambda^{\ast}(\nabla f(\Lambda(X)))
    = -\frac{1}{2}
    \sqrt{P} \Bigl(\sqrt{P} X \sqrt{P} \Bigr)^{-\frac{1}{2}} \sqrt{P},
  \end{equation}
  provided that 
  \begin{equation}
    \sqrt{P} X \sqrt{P} \in \Pd(\calY),
  \end{equation}
  which is equivalent to $X \in \Pd(\calX)$.
  In the case when $X$ is positive definite, one therefore has that
  \begin{equation}
   \nabla g(X)= -\frac{1}{2}
   \sqrt{P} \Bigl(\sqrt{P^{-1}} X^{-1} \sqrt{P^{-1}} \Bigr)^{\frac{1}{2}}
   \sqrt{P} = -\frac{1}{2} \bigl(P\,\#\,X^{-1}\bigr),
  \end{equation}
  as desired.

  Now suppose that $X\in\Pos(\calX)$ is not positive definite, and let
  $\Delta$ be the projection onto the kernel of $\sqrt{P} X \sqrt{P}$, which is
  nonzero by the assumption that $X$ is not positive definite.
  Consider the operator 
  \begin{equation}
    \label{eq:Y}
    Y = X + \lambda P^{-\frac{1}{2}} \Delta P^{-\frac{1}{2}}
  \end{equation}
  for an arbitrary choice of $\lambda > 0$.
  It is the case that
  \begin{equation}
    g(Y) - g(X) = 
    \Tr \sqrt{ \sqrt{P} X \sqrt{P}}
    - \Tr \sqrt{ \sqrt{P} X \sqrt{P}
      + \lambda \Delta} = -\sqrt{\lambda} \Tr(\Delta).
  \end{equation}
  Thus, if there were to exist an element $Z\in\partial g(X)$, one would have
  \begin{equation}
    g(Y) - g(X) \geq \langle Z, Y - X \rangle
  \end{equation}
  for all $Y\in\dom(f)$, including the operator \eqref{eq:Y} for every
  $\lambda > 0$. 
  It would then follow that
  \begin{equation}
    \lambda \Bigl\langle Z,P^{-\frac{1}{2}} \Delta P^{-\frac{1}{2}}\Bigr\rangle
    = \langle Z, Y - X\rangle \leq g(Y) - g(X) = -\sqrt{\lambda} \Tr(\Delta),
  \end{equation}
  or equivalently
  \begin{equation}
    \frac{\Bigl\langle Z,P^{-\frac{1}{2}} \Delta P^{-\frac{1}{2}}\Bigr\rangle}{
      \Tr(\Delta)} \leq - \frac{1}{\sqrt{\lambda}},
  \end{equation}
  for every $\lambda > 0$, which is impossible given that the left-hand side is
  a finite value independent of $\lambda$ and the right-hand side approaches
  $-\infty$ as $\lambda$ approaches 0.
\end{proof}

\subsubsection*{Gradients of functions involving the quantum relative entropy}

Let $\calY$ be a complex Euclidean space, let $P\in\Pd(\calY)$ be a positive
definite operator, and consider the function
$f:\Herm(\calY)\rightarrow\RR\cup\{\infty\}$ defined as
\begin{equation}
  \label{eq:f-for-relative-entropy}
  f(Y) = \begin{cases}
    \D(P \| Y) & \text{if}\; Y\in\Pd(\calY)\\
    \infty & \text{otherwise.}
  \end{cases}
\end{equation}
This function is differentiable at every positive definite operator
$Y\in\Pd(\calY)$, with its gradient being
\begin{equation}
  \nabla f(Y) = -D\log (Y)(P),
\end{equation}
where $D\log (Y)$ is the derivative of the logarithm function at $Y$.
If
\begin{equation}
  Y = \sum_{k=1}^m \lambda_k \Pi_k
\end{equation}
is the spectral decomposition of $Y$, then by means of the expression
\eqref{eq:derivative-from-divided-difference} this function can be described
explicitly as
\begin{equation}
\label{eq:DlogYZ}
  D\log (Y)(Z)
  = \sum_{k=1}^m \frac{1}{\lambda_k} \Pi_k Z \Pi_k
  + \sum_{j\not=k} \frac{\log(\lambda_j) - \log(\lambda_k)}{\lambda_j -
    \lambda_k} \Pi_j Z \Pi_k.
\end{equation}


\section{Proof of multiplicativity of fidelity of recovery}

\label{app:multiplicativity}

In this appendix, we illustrate how Theorem \ref{thm:cvxtransformationfids} may
be used to prove that the generalized fidelity of recovery is multiplicative.
This fact was first proved in \cite{BertaT16}, through the expression of the
fidelity of recovery as a semidefinite program.

Let $\rho_0\in\D(\calX_0\otimes\calZ_0)$, $\rho_1\in\D(\calX_1\otimes\calZ_1)$, $\sigma_0\in\D(\calY_0\otimes\calZ_0)$, and $\sigma_1\in\D(\calY_1\otimes\calZ_1)$ be density operators for some complex Euclidean spaces $\calX_0,\calX_1,\calY_0,\calY_1,\calZ_0,\calZ_1$, and consider the following pair of optimization problems:
\begin{equation}\label{eq:FoRpair}
  \begin{aligned}
    &\text{maximize}  &&
    \F\bigl(\sigma_0,\bigl(\Phi\otimes\I_{\L(\calZ_0)}\bigr)(\rho_0)\bigr)
    &\qquad&\text{maximize}  &&
    \F\bigl(\sigma_1,\bigl(\Phi\otimes\I_{\L(\calZ_1)}\bigr)(\rho_1)\bigr)
    \\    
    & \text{subject to} && \Phi\in\C(\calX_0,\calY_0)
    &     & \text{subject to} && \Phi\in\C(\calX_1,\calY_1).
  \end{aligned}
\end{equation}
Define density operators $\rho\in\D(\calX_0\otimes\calX_1\otimes\calZ_0\otimes\calZ_1)$ and $\sigma\in\D(\calY_0\otimes\calY_1\otimes\calZ_0\otimes\calZ_1)$ as
\begin{equation}
 \rho = W(\rho_0\otimes\rho_1)W^* 
 \qquad \text{and}\qquad
 \sigma = V(\sigma_0\otimes\sigma_1)V^* 
\end{equation}
where 
\begin{align*}
 W&\in\L(\calX_0\otimes\calZ_0\otimes\calX_1\otimes\calZ_1,\calX_0\otimes\calX_1\otimes\calZ_0\otimes\calZ_1)\\ \text{and}\quad V&\in\L(\calY_0\otimes\calZ_0\otimes\calY_1\otimes\calZ_1,\calY_0\otimes\calY_1\otimes\calZ_0\otimes\calZ_1)
\end{align*}
 are the isometries defined as
\begin{align*}
 W(x_0\otimes z_0\otimes x_1\otimes z_1) &= x_0\otimes x_1\otimes z_0\otimes z_1
 \\ \text{and}\quad
 V(y_0\otimes z_0\otimes y_1\otimes z_1) &= y_0\otimes y_1\otimes z_0\otimes z_1
\end{align*}
for all choices of vectors $x_0\in\calX_0$, $x_1\in\calX_1$, $y_0\in\calY_0$, $y_1\in\calY_1$, $z_0\in\calZ_0$, and $z_1\in\calZ_1$. Consider now the following optimization problem:
\begin{equation}\label{eq:FoRproduct}
  \begin{aligned}
    &\text{maximize}  &&
    \F\bigl(\sigma,\bigl(\Phi\otimes\I_{\L(\calZ_0\otimes\calZ_1)}\bigr)(\rho)\bigr)
    \\    
    & \text{subject to} && \Phi\in\C(\calX_0\otimes\calX_1,\calY_0\otimes\calY_1).
  \end{aligned}
\end{equation}
The fact that the generalized fidelity of recovery is multiplicative can be stated as follows.
Let $\Phi_0\in\C(\calX_0,\calY_0)$ and $\Phi_1\in\C(\calX_1,\calY_1)$ be channels and suppose that this pair of channels is optimal for the pair of optimization problems in \eqref{eq:FoRpair}. Then the channel $\Phi_0\otimes\Phi_1\in\C(\calX_0\otimes\calX_1,\calY_0\otimes\calY_1)$ is optimal for the optimization problem in \eqref{eq:FoRproduct}. To prove this fact, we may define operators $H_0\in\Herm(\calY_0\otimes\calX_0)$ and $H_1\in\Herm(\calY_1\otimes\calX_1)$ as
\begin{align*}
 H_0&=-\frac{1}{2}(\I_{\L(\calY_0)}\otimes \Psi_{\rho_0}^{\ast})
      \Bigl(
      \sqrt{\sigma_0}\tinyspace
      \bigl( \sqrt{\sigma_0}\, (\Phi_0\otimes\I_{\L(\calZ_0)})(\rho_0)
      \sqrt{\sigma_0}\bigr)^{-\frac{1}{2}}
      \sqrt{\sigma_0}
      \Bigr)\\
 \text{and}\quad
 H_1&=-\frac{1}{2}(\I_{\L(\calY_1)}\otimes \Psi_{\rho_1}^{\ast})
      \Bigl(
      \sqrt{\sigma_1}\tinyspace
      \bigl( \sqrt{\sigma_1}\, (\Phi_1\otimes\I_{\L(\calZ_1)})(\rho_1)
      \sqrt{\sigma_1}\bigr)^{-\frac{1}{2}}
      \sqrt{\sigma_1}
      \Bigr).
\end{align*}
From the assumption that both $\Phi_0$ and $\Phi_1$ are optimal, it follows from Theorem \ref{thm:cvxtransformationfids} that the following conditions hold:
\begin{enumerate}
 \item $\im(\sigma_0) \subseteq \im\bigl(\bigl(\Phi_0\otimes\I_{\L(\calZ_0)}\bigr)(\rho_0) \bigr)$ and $\im(\sigma_1) \subseteq \im\bigl(\bigl(\Phi_1\otimes\I_{\L(\calZ_1)}\bigr)(\rho_1) \bigr)$
 \item $\Tr_{\calY_0}(H_0 J(\Phi_0)) \in \Herm(\calX_0)$ and  $\Tr_{\calY_1}(H_1 J(\Phi_1)) \in \Herm(\calX_1)$
 \item  $H_0 \geq \I_{\calY_0}\otimes\Tr_{\calY_0}(H_0 J(\Phi_0))$ and $H_1 \geq \I_{\calY_1}\otimes\Tr_{\calY_1}(H_1 J(\Phi_1))$.
\end{enumerate}
It is evident that $\im(\sigma_0\otimes\sigma_1) \subseteq \im\bigl(\bigl(\Phi_0\otimes\I_{\L(\calZ_0)}\otimes\Phi_1\otimes\I_{\L(\calZ_1)}\bigr)(\rho_0\otimes\rho_1) \bigr)$ by taking tensor products, and it follows that
\begin{equation}
 \im(\sigma) \subseteq\im\bigl(\bigl(\Phi_0\otimes\Phi_1\otimes\I_{\L(\calZ_0\otimes\calZ_1)}\bigr)(\rho)\bigr))
\end{equation}
by rearranging the spaces. Define the operator
\begin{equation}
H = -\frac{1}{2}(\I_{\L(\calY_0\otimes\calY_1)}\otimes \Psi_{\rho}^{\ast})
      \Bigl(
      \sqrt{\sigma}\tinyspace
      \bigl( \sqrt{\sigma}\, (\Phi_0\otimes\Phi_1\otimes\I_{\L(\calZ_0\otimes\calZ_1)})(\rho)
      \sqrt{\sigma}\bigr)^{-\frac{1}{2}}
      \sqrt{\sigma}
      \Bigr),
\end{equation}
and define the isometry $U\in\L(\calY_0\otimes\calX_0\otimes\calY_1\otimes\calX_1,\calY_0\otimes\calY_1\otimes\calX_0\otimes\calX_1)$ as the operator satisfying
\begin{equation}
 U(y_0\otimes x_0\otimes y_1\otimes x_1) = y_0\otimes y_1\otimes x_0\otimes x_1
\end{equation}
for every choice of vectors $x_0\in\calX_0$, $x_1\in\calX_1$, $y_0\in\calY_0$, and $y_1\in\calY_1$. With these choices of operators, it is evident that $H = U(H_0\otimes H_1)U^*$ and it follows that
\begin{equation}
 \Tr_{\calY_0\otimes\calY_1}(HJ(\Phi_0\otimes\Phi_1))  = \Tr_{\calY_0}(H_0 J(\Phi_0))\otimes\Tr_{\calY_1}(H_1 J(\Phi_1)) \in \Herm(\calX_0\otimes\calX_1),
\end{equation}
as the tensor product of Hermitian operators is Hermitian. Moreover it holds that
\begin{align*}
 H & = U(H_0\otimes H_1)U^*\\
   & \geq U\Bigl(\bigl(\I_{\calY_0}\otimes\Tr_{\calY_0}(H_0 J(\Phi_0))\bigr)\otimes\bigl(\I_{\calY_1}\otimes\Tr_{\calY_1}(H_1 J(\Phi_1))\bigr)\Bigr)U^*\\
    &= \I_{\calY_0\otimes\calY_1}\otimes\Tr_{\calY_0\otimes\calY_1}(H J(\Phi_0\otimes\Phi_1)),
\end{align*}
which follows from the fact that $P_0\otimes P_1\geq Q_0\otimes Q_1$ holds for every choice of positive semidefinite operators $P_0,P_1,Q_0,Q_1$ satisfying $P_0\geq Q_0$ and $P_1\geq Q_1$. As the conditions of Theorem \ref{thm:cvxtransformationfids} are satisfied for the optimization problem in \eqref{eq:FoRproduct}, one has that $\Phi_0\otimes\Phi_1$ is optimal for this problem.

\bibliographystyle{plainnat}


\end{document}